%% file: NECO-07-17-2918-Source.tex
\documentclass[12pt]{article} 
\usepackage[utf8]{inputenc}
\pdfoutput=1

\usepackage{amsmath}
\usepackage{times}
\usepackage{graphicx}
\usepackage{color}
\usepackage{multirow}

\usepackage{rotating}
\usepackage{bbm}
\usepackage{latexsym}

\usepackage[hidelinks,breaklinks=true]{hyperref}
\usepackage{url}

\usepackage{setspace}
\usepackage[skip=6pt, font=doublespacing]{caption}
\usepackage[labelfont=bf]{subcaption}
\DeclareCaptionLabelFormat{bold}{\textbf{(#2)}}
\usepackage{algorithm}
\usepackage{algorithmic}
\usepackage{hyperref}

\usepackage{NECO-07-17-2918-Source-mystyle}
\usepackage{NECO-07-17-2918-Source-notations}

\graphicspath{{NECO-07-17-2918-Source-figures/}}
\begin{document}

\TitleNECO{Bayesian Modeling of Motion Perception using Dynamical Stochastic Textures}

\AuthorNECO{Jonathan Vacher}{1,4,5},  \quad 
\AuthorNECO{Andrew Isaac Meso}{2,3}, \\
\AuthorNECO{Laurent U. Perrinet}{2,5} 
 \quad {\bf \large and}  \quad 
\AuthorNECO{Gabriel Peyr\'e}{1,5}
\vspace{5mm}\\
\AffilNECO{1}{D\'epartement de Math\'ematiques et Applications,\\ \'Ecole Normale Sup\'erieure, Paris, FRANCE}
\AffilNECO{2}{Institut de Neurosciences de la Timone, Marseille, FRANCE}
\AffilNECO{3}{Faculty of Science and Technology, Bournemouth University,\\ Poole, UNITED KINGDOM}
\AffilNECO{4}{UNIC, Gif-sur-Yvette, FRANCE}
\AffilNECO{5}{CNRS, FRANCE}
\KeywordsNECO{Dynamic textures, Motion perception, Bayesian Modelling, Stochastic Partial Differential Equations, Psychophysics}

\ShortTitleNECO{Biologically Inspired Dynamic Textures}

\input{NECO-07-17-2918-Source-sections/NECO-07-17-2918-Source-sec-abstract}
\input{NECO-07-17-2918-Source-sections/NECO-07-17-2918-Source-sec-intro}
\input{NECO-07-17-2918-Source-sections/NECO-07-17-2918-Source-sec-mc-definition}

\input{NECO-07-17-2918-Source-sections/NECO-07-17-2918-Source-sec-mc-spde}
\input{NECO-07-17-2918-Source-sections/NECO-07-17-2918-Source-sec-psychophysics}
\input{NECO-07-17-2918-Source-sections/NECO-07-17-2918-Source-sec-discussion}
\printbibliography[heading=subbibliography]
\input{NECO-07-17-2918-Source-sections/NECO-07-17-2918-Source-sec-appendix}

\end{document}

%% file: NECO-07-17-2918-Source-sections/NECO-07-17-2918-Source-sec-abstract.tex

\begin{center} {\bf Abstract} \end{center}

A common practice to account for psychophysical biases in vision is to frame them as consequences of a dynamic process relying on optimal inference with respect to a generative model. The present study details the complete formulation of such a generative model intended to probe visual motion perception with a dynamic texture model.
It is first derived in a set of axiomatic steps constrained by biological plausibility. 
We extend previous contributions by detailing three equivalent formulations of this texture model. First, the composite dynamic textures are constructed by the random aggregation of warped patterns, which can be viewed as 3D Gaussian fields. Secondly, these textures are cast as solutions to a stochastic partial differential equation (sPDE). This essential step enables real time, on-the-fly texture synthesis using time-discretized auto-regressive processes. 
It also allows for the derivation of a local motion-energy model, which corresponds to the log-likelihood of the probability density. The log-likelihoods are  essential for the construction of a Bayesian inference framework. 
We use the dynamic texture model to psychophysically probe speed perception in humans using zoom-like changes in the spatial frequency content of the stimulus. 
The human data replicates previous findings showing perceived speed to be positively biased by spatial frequency increments.
A Bayesian observer who combines a Gaussian likelihood centered at the true speed and a spatial frequency dependent width with a ``slow speed prior'' successfully accounts for the perceptual bias. More precisely, the bias arises from a decrease in the observer's likelihood width estimated from the experiments as the spatial frequency increases. Such a trend is compatible with the trend of the dynamic texture likelihood width.

%% file: NECO-07-17-2918-Source-sections/NECO-07-17-2918-Source-sec-intro.tex
\section{Introduction}

\subsection{Modeling visual motion perception }

A normative explanation for the function of perception is to infer relevant unknown real world parameters from the sensory input with respect to a generative model~\citep{Gregory80}. Equipped with some prior knowledge about both the nature of neural representations and the structure of the world, the modeling approach that emerges corresponds to the \emph{Bayesian brain} hypothesis~\citep{knill2004bayesian,doya2007bayesian,colombo2012bayes,kersten2004object}. This assumes that when given some sensory information $S$, the brain uses neural computations which ultimately conform with Bayes' theorem :
\eql{\label{eq-bayes}
\PP_{M|S}(m|s) = \frac{\PP_{S|M}(s|m)\PP_M(m)}{\PP_S(s)}.
}
This computation yields an estimate of the parameters $m$ where the probability distribution function $\PP_{S|M}$ is given by the generative model and $ \PP_M$ represents prior knowledge.
This hypothesis has been well illustrated with the case of motion perception~\citep{Weiss02}. This framework uses a Gaussian parameterization of the generative model and a unimodal (Gaussian) prior in order to estimate perceived speed $v$ when observing a visual input $I$. 

However, Gaussian likelihoods and priors do not always fit with psychophysical results~\citep{Wei12,hassan2015perceptual}. As such, a major challenge is to refine the construction of generative models so that they are consistent with the widest variety of empirical results.  

In fact, the estimation problem inherent to perception is successfully solved, in part, through the definition of an adequate generative model. Probably the simplest generative model to describe visual motion is the luminance conservation equation~\citep{Adelson85}. It states that luminance $I(x,t)$ for $(x,t) \in \RR^2\times \RR$ is approximately conserved along trajectories defined as integral lines of a vector field $v(x,t) \in \RR^2\times \RR$. The corresponding generative model defines random fields as solutions to the stochastic partial differential equation (sPDE), 
\eql{\label{eq-luminance}
	\dotp{v}{\nabla I} + \pd{I}{t}  = W,
} 
where $\dotp{\cdot}{\cdot}$ denotes the Euclidean scalar product in $\RR^2$, $\nabla I$ is the spatial gradient of~$I$. To match the distribution of spatial scale statistics of natural scenes (\textit{ie} the 1/f amplitude fall-off of spatial frequencies) or some alternative category of textures, the driving term $W$ is usually defined as a stationary colored Gaussian noise corresponding to the average localized spatio-temporal correlation (which we refer to as the spatio-temporal coupling), and is parameterized by a covariance matrix $\Sigma$, while the field is usually a constant vector $v(x, t)=v_0$ accounting for a full-field translation with constant speed. 

Ultimately, the application of this generative model is useful for probing the visual system with a probabilistic approach, for instance for one seeking to understand how observers might detect motion in a scene.  %
Indeed, as shown by~\textcite{Nestares00,Weiss02}, the negative log-likelihood of the probability distribution of the solutions $I$ to the luminance conservation equation~\eqref{eq-luminance} (on domain $\Om \times [0,T]$ and for constant speed $v(x,t)=v_0$) is proportional to the value of the motion-energy model~\citep{Adelson85} given by 
\eql{\label{energy}
	 \int_\Om \int_0^T |\dotp{v_0}{\nabla (K \star I)(x,t)} + \pd{(K \star I)}{t}(x,t)  |^2 \d t \, \d x
} 
where $K$ is the whitening filter corresponding to the inverse square root of $\Si$, and $\star$ is the convolution operator. Using some prior knowledge about the expected distribution of motions, for instance a preference for slow speeds, a Bayesian formalization can be applied to this inference problem~\citep{Weiss01,Weiss02}. 

\subsection{Previous Works in Context}

\paragraph{Dynamic Texture Synthesis.}

The model defined in equation~\eqref{eq-luminance} is quite simplistic compared to the complexity of natural scenes. It is therefore useful here to discuss generative models associated with texture synthesis methods previously proposed in the computer vision and computer graphics community. 
Indeed, the literature on the subject of static textures synthesis is abundant (see for instance~\citep{WLKT09}). %
Of particular interest for us is the work by~\textcite{Galerne11,GalernePHD}, which proposes a stationary Gaussian model restricted to static textures. This provides an equivalent generative model based on Poisson shot noise. %
Realistic dynamic texture models have received less attention, and the most prominent method is the non-parametric Gaussian auto-regressive (AR) framework developed by~\textcite{Doretto-Dyntex}, which has been thoroughly explored~\citep{2014-xia-siims,yuan2004synthesizing,costantini2008higher,filip2006fast,hyndman2007higher,abraham2005dynamic}. These works generally consist in finding an appropriate low-dimensional feature space in which an AR process models the dynamics. Many of these approaches focus on the feature space where the decomposition is efficiently performed using Singular Value Decomposition (SVD) or its Higher Order version (HOSVD)~\citep{Doretto-Dyntex,costantini2008higher}. In~\textcite{abraham2005dynamic}, the feature space is the Fourier frequency domain, and the AR recursion is carried independently over each frequency, which defines the space-time stationary processes. 
A similar approach is used in~\citep{2014-xia-siims} to compute the average of several dynamic texture models. 
Properties of these AR models have been studied by~\textcite{hyndman2007higher} who find that higher order AR processes are able to capture perceptible temporal features. A different approach aims at learning the manifold structure of a given dynamic texture~\citep{liu2006dynamic} while yet another deals with motion statistics~\citep{rahman2008dynamic}. What all these works have in common is the aim to reproduce the natural spatio-temporal behavior of dynamic textures with rigorous mathematical tools. Similarly, our concern is to design a dynamic texture model that is precisely parametrized for experimental purposes in visual neuroscience and psychophysics.

\paragraph{Stochastic Differential Equations (sODE and sPDE).}

Stochastic Ordinary Differential Equations (sODE) and their higher dimensional counter-parts, stochastic Partial Differential Equations (sPDE) can be viewed as continuous-time versions of these 1-D or higher dimensional auto-regressive (AR) models. Conversely, AR processes can therefore also be used to compute numerical solutions to these sPDE using finite difference approximations of time derivatives.  
Informally, these equations can be understood as partial differential equations perturbed by a random noise. The theoretical and numerical study of these sDE is of fundamental interest in fields as diverse as physics and chemistry~\citep{van1992stochastic}, finance~\citep{el1997backward} or neuroscience~\citep{fox1997stochastic}. They allow for the dynamic study of complex, irregular and random phenomena such as particle interactions, stock or saving prices, or ensembles of neurons. 
In psychophysics, sODE have been used to model  decision making tasks in which the stochastic variable represents the accumulation of knowledge until the decision is made, thus providing detailed information about predicted response times \citep{smith2000stochastic}. 
In imaging sciences, sPDE with sparse non-Gaussian driving noise have been proposed as models of natural signals and images~\citep{UnserBook}. 
As described above, the simple motion energy model~\eqref{energy} can similarly be demonstrated to rely on the sPDE equation~\eqref{eq-luminance} of a stochastic model of visual sensory input. This has not previously been presented in a formal way in the literature. One key goal of this paper is to comprehensively formulate a parametric family of Gaussian sPDEs which describes the modeling of moving images (and the corresponding synthesis of visual stimulation) and thus allow for a fine-grained systematic exploration of psychophysical behavior. 

\paragraph{Inverse Bayesian inference.}

Importantly, these dynamic stochastic models are closely related to the likelihood and prior models which serve to infer motion estimates from the dynamic visual stimulation. In order to account for perceptual bias, a now well-accepted methodology in the field of psychophysics is to assume that observers are ``ideal observers'' and therefore make decisions using optimal statistical inference, typically a maximum-a-posteriori (MAP) estimator, using Bayes' formula to combine this likelihood with some internal prior, see equation~\eqref{eq-bayes}. 
Several experimental studies use this hypothesis as a justification for the observed perceptual biases by proposing some adjusted likelihood and prior models \citep{doya2007bayesian,colombo2012bayes}, and more recent works push these ideas even further. Observing some perceptual bias, is it possible to ``invert'' this forward Bayesian decision-making process, and infer the (unknown) internal prior that best fits a set of observed experimental choices made by observers?
Following~\textcite{Stocker06}, we coined this promising methodology ``inverse Bayesian inference''. This is of course an ill-posed, and highly non-linear inverse problem, making it necessary to add constraints on both the prior and the likelihood to make it tractable. For instance~\textcite{SotiropoulosVR, Stocker06,jogan2015signal} impose smoothness constraints in order to be able to locally fit the slope of the prior.  
Herein, we propose to use visual stimulations generated by the (forward) generative model to test these ``inverse'' Bayesian models. To allow for a simple yet mathematically rigorous analysis of this approach within the context of speed discrimination, in the present study we will use a restricted parametric set of descriptors for the likelihood and priors. This provides a self-consistent approach to test the visual system, from stimulation to behavior analysis 

\subsection{Contributions}

In this paper, we lay the foundations that we hope will enable a better understanding of human motion perception by improving generative models for dynamic texture synthesis. From that perspective, we motivate the generation of visual stimulation within a stationary Gaussian dynamic texture model. 

We develop our current model by extending, mathematically detailing and testing in psychophysical experiments previously introduced dynamic noise textures~\citep{Leon12,Simoncini12,vacher2015biologically,gekas2017normalization} coined ``Motion Clouds'' (MC). 
Our first contribution is a complete axiomatic derivation of the model, seen as a shot noise aggregation of dynamically warped ``textons''.
Within our generative model, the parameters correspond to average spatial and temporal transformations (\textit{ie} zoom, orientation and translation speed) and associated standard deviations of random fluctuations, as illustrated in Figure~\ref{fig:warps}, with respect to external (objects) and internal (observer) movements. 
The second main contribution is the explicit demonstration of the equivalence between this model and a class of linear sPDEs. This shows that our model is a generalization of the well-known luminance conservation (see equation~\ref{eq-luminance}). 
This sPDE formulation has two chief advantages: it allows for a real-time synthesis using an AR recurrence (in the form of a GPU implementation) and allows one to recast the log-likelihood of the model as a generalization of the classical motion energy model, which in turn is crucial to allow for Bayesian modeling of perceptual biases. 
Our last contribution follows from the Bayesian approach and is an illustrative application of this model to the psychophysical study of motion perception in humans. 
This example of the model development constrains the likelihood, which in turn enables a simple fitting procedure to be performed using both an empirical and a larger Monte-Carlo derived synthetic dataset to determine the prior driving the perceptual biases. 
The code associated to this work is available at \url{https://github.com/JonathanVacher/projects/tree/master/bayesian_observer}.
\subsection{Notations}
\label{sec:notations}

In the following, we denote $(x,t) \in \RR^2 \times \RR$ the space/time variable, and $(\xi,\tau) \in \RR^2 \times \RR$ the corresponding frequency variables. If $f(x,t)$ is a function defined on $\RR^3$, then its Fourier transform is defined as
\eq{
	\hat f(\xi,\tau) \eqdef \int_{\RR^2} \int_{\RR} f(x,t) e^{-\imath (\dotp{x}{\xi} + \tau t)  } \d t \d x .
}
For $\xi \in \RR^2$, we denote $\xi = \norm{\xi}(\cos( \angle{\xi}),\sin(\angle{\xi})) \in \RR^2$ its polar coordinates. For a function $g$ defined on $\RR^2$, we denote $\bar g(x) = g(-x)$.  
In the following, we denote a random variable with a capital letter such as  $A$ and $a$ as a realization of $A$. We note as $\dis_A(a)$ the corresponding probability distribution of $A$. 

%% file: NECO-07-17-2918-Source-sections/NECO-07-17-2918-Source-sec-mc-definition.tex
\section{Axiomatic Construction of the Dynamic Textures}
\label{sec-axiomatic}
Dynamic textures which are efficient to probe visual perception should be generated from low-dimensional yet naturalistic parametric stochastic models. They should embed meaningful physical parameters (such as the effect of head rotations or whole-field scene movements, see Figure~\ref{fig:warps}) into the local or global dependencies of the random field (for instance the covariance). 
In the luminance conservation model~\eqref{eq-luminance}, the generative model is parameterized by a spatio-temporal coupling encoded in the covariance $\Sigma$ of the driving noise and the motion flow $\vz$.

This localized space-time coupling (e.g. the covariance, if one restricts one's attention to Gaussian fields) is essential, as it quantifies the extent of the spatial integration area as well as the integration dynamics. This is an important issue in neuroscience, when considering the implementation of spatio-temporal integration mechanisms from very small to very large scales, i.e. going from ‘local’ to ‘global’ visual features~\citep{Rousselet2004,Born2005,DiCarlo2012}.
In particular, this is crucial to understand the modular sensitivity within the different lower visual areas. In primates for instance, the Primary Visual Cortex (V1) generally encodes small features in a given range of spatio-temporal scales. In contrast, ascending the processing hierarchy, the Middle Temporal (V5/MT) area exhibits selectivity for larger visual features.
For instance, by varying the spatial frequency bandwidth of such dynamic textures, distinct mechanisms for perception and action have been identified in humans~\citep{Simoncini12}. Our goal here is to develop a principled axiomatic definition of these dynamic textures.

\begin{figure}[b!]
\vspace{2mm}
\begin{center}
\includegraphics[width=.85\textwidth]{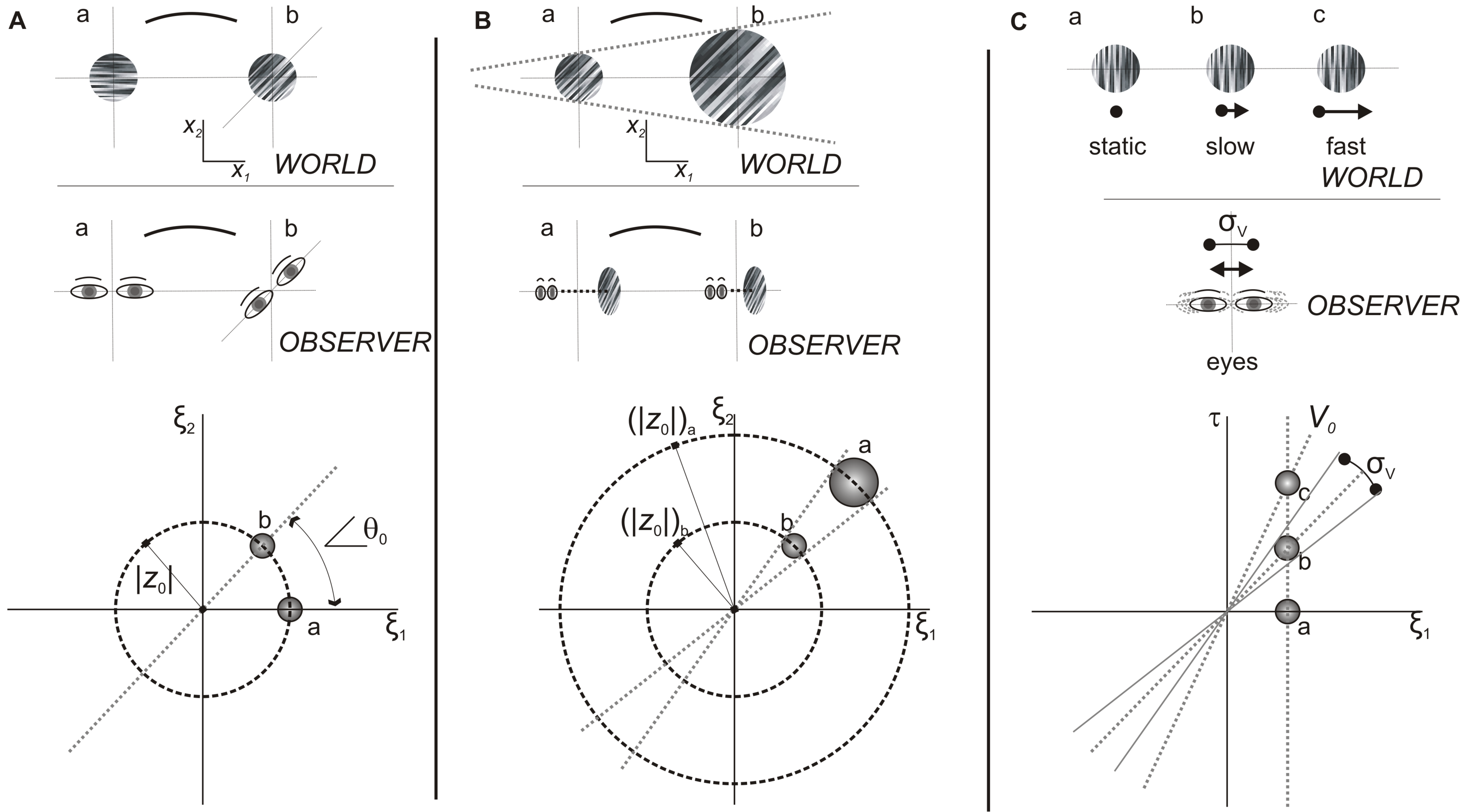}
\end{center}
\caption{\emph{Parameterization of the class of Motion Clouds (MC) stimuli.} The illustration relates the parametric changes in MC with real world (top row) and observer (second row) movements. \textbf{(A)}~Orientation changes resulting in scene rotation are parameterized through $\th$, as shown in the bottom row where $(a)$~horizontal and $(b)$~obliquely oriented MCs are compared. \textbf{(B)}~Zoom movements, either from scene looming or observer movements in depth, are characterized by scale changes reflected by a frequency term $\z$ shown for $(a)$~a more distant viewpoint compared to $(b)$~a closer one. 
\textbf{(C)}~Translational movements in the scene characterized by $V$ using the same formulation for $(a)$~static, $(b)$~slow and $(c)$~fast moving MC, with the variability in these speeds quantified by $\sr$. The variables $\xi$ and $\tau$ in the third row are the spatial and temporal frequency scale parameters. The development of this formulation is detailed in the text. }
\label{fig:warps}
\end{figure}

\subsection{From Shot Noise to Motion Clouds}

We propose a derivation of a general parametric model of dynamic textures. This model is defined by aggregation, through summation, of a basic spatial ``texton'' template $g(x)$. The summation reflects a transparency hypothesis, which has been adopted for instance by~\textcite{Galerne11}. While one could argue that this hypothesis is overly simplistic and does not model occlusions or edges, it leads to a tractable framework of stationary Gaussian textures, which has proved useful to model static micro-textures~\citep{Galerne11} and dynamic natural phenomena~\citep{2014-xia-siims}. The simplicity of this framework allows for a fine tuning of frequency-based (Fourier) parameterization, which is desirable for the interpretation of psychophysical experiments with respect to underlying spatio-temporal neural sensitivity.

We define a random field as 
\eql{\label{eq-deadleaves}
	I_\la(x,t) \eqdef  \frac{1}{\sqrt{\la}} \sum_{p \in \NN} g( \phi_{\Geom_p}( x-X_{p}  - V_p t ) )
}
where $\phi_\geom : \RR^2 \rightarrow \RR^2$ is a planar deformation parameterized by a finite dimensional vector $\geom$. The parameters $(X_p,V_p,\Geom_p)_{p \in \NN}$ are independent and identically distributed random vectors. They account for the variability in the position of objects or observers ($\phi_{A_p}$) and their speed ($V_p$), thus mimicking natural motions in an ambient scene. The set of translations $(X_p)_{p \in \NN}$ is a 2-D Poisson point process of intensity $\la>0$.
This means that, defining for any measurable $A$, $C(A) = \sharp\enscond{p}{X_p \in A}$, $C(A)$ has a Poisson distribution with mean $\la |A|$ (where $|A|$ is the measure of $A$) and $C(A)$ is independent of $C(B)$ if $A \cap B = \emptyset$. 

Intuitively, this model~\eqref{eq-deadleaves} corresponds to a dense mixing of stereotyped static textons as in the work of~\textcite{Galerne11}. In addition to the extension to the temporal domain, the originality of our approach is two-fold. First, the components of this mixing are derived from the texton by visual transformations $\phi_{\Geom_p}$, which may correspond to arbitrary transformations such as zooms and/or rotations (in which case $\Geom_p$ is a vector containing the scale factor and the rotation angle). See the illustration in Figure~\ref{fig:warps}. Second, we explicitly model the motion (position $X_p$ and speed $V_p$) of each individual texton. 

In the following, we denote $\falpha$ the common distribution of the i.i.d. $(\Geom_p)_p$, and we denote $\fv$ the distribution in $\RR^2$ of the speed vectors $(V_p)_p$.  Section~\ref{sec-bio-inspired} instantiates this model and proposes canonical choices for these variabilities. 

The following result shows that the model~\eqref{eq-deadleaves} converges for high point density $\la \rightarrow +\infty$ to a stationary Gaussian field and gives the parameterization of the covariance.
Its proof follows from a specialization of Theorem 3.1 in~\citep{GalernePHD} to our setting. 

\begin{prop}\label{eq-conv-deadleaves}
	$I_\la$ is stationary with bounded second order moments. Its covariance is $\Si(x,t,x',t') = \ga(x-x',t-t')$ where $\ga$ satisfies 
	\eql{\label{eq-cov-proposition}
		\foralls (x,t) \in \RR^3, \quad
		\ga(x,t) = \iiint_{\RR^2} c_g(\phi_\geom(x-\speed t))  \fv(\speed) \falpha(\geom) \d \speed \d \geom
	} 
	where $c_g = g \star \bar g$ is the auto-correlation of $g$. 
	When $\la \rightarrow +\infty$, it converges (in the sense of finite dimensional distributions) toward a stationary Gaussian field $I$ of zero mean and covariance $\Si$. 
\end{prop}

This proposition enables us to give a precise definition of a MC. 

\begin{defn}\label{defn-mc-fourier} 
	A Motion Cloud (MC) is a stationary Gaussian field whose covariance is given by equation~\eqref{eq-cov-proposition}. 
\end{defn}

Note that, following~\textcite{Galerne10}, the convergence result of Proposition~\ref{eq-conv-deadleaves} could be used in practice to simulate a Motion Cloud  $I$ using a high but finite value of $\la$ in order to generate a realization of $I_\la$. We do not use this approach, and rather rely on the sPDE characterization proved in Section~\ref{sec-spde-model}, which is well tailored for an accurate and computationally efficient dynamic synthesis.


\subsection{Towards ``Motion Clouds'' for Experimental Purposes}
The previous Section provides a theoretical definition of MC (see Definition~\ref{defn-mc-fourier}) that is characterized by $c_g, \phi_\geom, \falpha$ and $\fv$. In order to have better control of the covariance $\ga$ one needs to resort to a low-dimensional representation of these parameters. We further study this model in the specific case where the warps $\phi_a$ are rotations and scalings (see Figure~\ref{fig:warps}). They account for the characteristic orientations and sizes (or spatial scales) of a scene, in relation to the observer. We thus set
\eql{\label{eq:warping}
	\foralls \geom = (\th,\z) \in [-\pi, \pi) \times \RR_+^{*}, \quad
	\phi_\geom(x) \eqdef \z R_{-\th}(x),  
}
where $R_\th$ is the planar rotation of angle $\th$. We now give some physical and biological motivation to account for our particular choices for the distributions of the parameters. We assume that the distributions $\fz$ and $\ftheta$ of spatial scales $\z$ and orientations $\th$, respectively (see Figure~\ref{fig:warps}), are independent and have densities, thus considering 
\eql{\label{eq:warp-dis}
	\foralls \geom = (\th,\z) \in [-\pi, \pi) \times \RR_+^{*}, \quad
	\falpha(\geom) = \fz(\z) \, \ftheta(\th).
 }
The speed vector $\speed$ is assumed to be randomly fluctuating around a central speed $\vz \in \RR^2$, so that
\eql{\label{eq-distr-V}
	\foralls \speed \in \RR^2, \quad \fv(\speed) = \fr(\norm{\speed - \vz}).
}
In order to obtain ``optimal'' responses to the stimulation (as advocated by~\textcite{young01}) and based on the structure of a standard receptive field of V1, it makes sense to define the texton so that it resembles an oriented Gabor~\citep{Fischer07cv}. Such an elementary luminance feature acts as the generic atom
\eql{\label{eq-texton}
g_\sigma(x) = \frac{1}{2 \pi} \cos\left( \dotp{x}{\xiz} \right) e^{ - \frac{\sigma^2}{2} \norm{x}^2 }
}
where $\sigma$ is the inverse of the standard deviation and $\xiz \in \RR^2$ is the spatial frequency. Since the orientation and scale of the texton is handled by the $(\th,\z)$ parameters, we can impose the normalization $\xiz = (1,0)$ without loss of generality. 
In the special case where $\sigma \rightarrow 0$, $g_\sigma$ is a grating of frequency $\xiz$, and the image $I$ is a dense mixture of drifting gratings, whose power-spectrum has a closed form expression detailed in Proposition~\ref{prop-mc-spectrum}. It is fully parameterized by the distributions $(\fz,\ftheta,\fv)$ and the central frequency and speed $(\xiz,\vz)$. 
Note that it is possible to consider any arbitrary textons $g$, which would give rise to more complicated parameterizations for the power spectrum $\hat g$, but here we decided to stick to the simple asymptotic case of gratings. 


\begin{prop}\label{prop-mc-spectrum}
	Consider the texton $g_\sigma$ , when $\sigma \rightarrow 0$,
	the Gaussian field $I_{\sigma}(x,t)$ defined in Proposition~\ref{eq-conv-deadleaves} converges toward a stationary Gaussian field of covariance having the power-spectrum
	\eql{\label{eq-dfn-mc-spectrum}
		\foralls (\xi,\tau)  \in \RR^2 \times \RR, \:
		\hat \ga(\xi,\tau) = \frac{ \fz\pa{  \norm{\xi} } }{\norm{\xi}^2}
			\ftheta\pa{ \angle{\xi} }
			\Ll(\fr)
			\pa{ 
				-\frac{
					\tau + \dotp{\vz}{\xi}
				}{  
					\norm{\xi} 
				}
			}, 
	}
where the linear transform $\Ll$ is such that 
\eql{\label{eq:linear-trans-spe}
	\foralls u \in \RR,  \quad
	\Ll(f)(u) \eqdef \int_{-\pi}^{\pi} f( -u/ \cos(\phi) ) \d \phi.
}  	
\end{prop}
\begin{proof}
See Appendix~\ref{proof:prop-mc-spectrum}.
\end{proof}

\begin{rem}
Note that the envelope of $\hat\ga$ as defined in equation~\eqref{eq-dfn-mc-spectrum} is constrained to lie within a cone in the spatio-temporal domain with the apex at zero (see the division by $\norm{\xi}$ in the argument of $\Ll(\fr)$). This is an important and novel contribution, when compared to a classical Gabor. Basing the generation of the textures on distributions of translations, rotations and zooms, we provide a principled approach to show that speed bandwidth gets scaled with spatial frequency to provide a scale invariant model of moving texture transformations.
\end{rem}

\subsection{Biologically-inspired Parameter Distributions}
\label{sec-bio-inspired}

We now give meaningful specialization for the probability distributions $\fz$, $\ftheta$, and $\fr$, which are inspired by some known scaling properties of the visual transformations relevant to dynamic scene perception. 
\paragraph{Parameterization of $\fz$.}
First, the observer's small, centered linear movements along the axis of view (orthogonal to the plane of the scene) generate centered planar zooms of the image. From  the linear modeling of the observer's displacement and the subsequent multiplicative nature of zoom, scaling should follow a Weber-Fechner law. This law states that subjective perceptual sensitivity when quantified is proportional to the logarithm of stimulus intensity. Thus, we choose the scaling $\z$ drawn from a log-normal distribution $\fz$, defined in equation~\eqref{eq-biosinspired-fz}. The parameter $\tlnsz$ quantifies the variation in the amplitude of zooms of individual textons relative to the characteristic scale $\tlnz$. We thus define
\eql{\label{eq-biosinspired-fz}
	\fz(\z) \propto \dfrac{\tlnz}{\z} 
	\exp \pa{
		-\frac{\ln\left( \frac{\z}{\tlnz} \right)^2}{2\ln\left(1 + \tlnsz^2 \right)}
	}, 
}
where $\propto$ means that we did not include the normalizing constant. In practice, we may prefer to parametrize this distribution by its mode and octave bandwidth $(z_0,B_z)$ instead of $(\tlnz,\tlnsz)$. See Appendix~\ref{sec:appendix} where we discuss two different parametrizations.

\paragraph{Parameterization of $\ftheta$.}

In our model, the texture is perturbed by variations in the global angle $\theta$ of the scene: for instance, the head of the observer may roll slightly around its normal position. The von-Mises distribution -- as a good approximation of the warped Gaussian distribution around the unit circle -- is an adapted choice for the distribution of $\theta$  with mean $\theta_0$ and bandwidth $\stheta$, 
\eql{\label{eq-biosinspired-ftheta}
	\ftheta(\th) \propto e^{ \frac{\cos (2 ( \theta - \theta_0 ))}{4 \stheta^2} }
}

\paragraph{Parameterization of $\fr$.}

We may similarly consider that the position of the observer is variable in time. On the first order approximation, movements perpendicular to the axis of view dominate, generating random perturbations to the global translation $\vz$ of the image at speed $\speed-\vz \in \RR^2$. These perturbations are for instance described by a Gaussian random walk: take for instance tremors, which are small, constant and jittering movements of the eye ($\leq 1$ deg). This justifies the choice of a radial distribution~\eqref{eq-distr-V} for $\fv$. 
This radial distribution $\fr$ is thus selected as a bell-shaped function of width $\sr$, and we choose here a Gaussian function for its generality
\eql{\label{eq-biosinspired-fr}
	\fr(r) \propto e^{-\frac{r^2}{ 2{\sr}^2 }}.
}
Note that, as detailed in Section~\ref{sec-spde}, a slightly different bell-function (with a more complicated expression) should be used to obtain an exact equivalence with the sPDE discretization.

\paragraph{Bringing everything together.}

\begin{figure}[!ht]
\begin{center}
\includegraphics[scale=1]{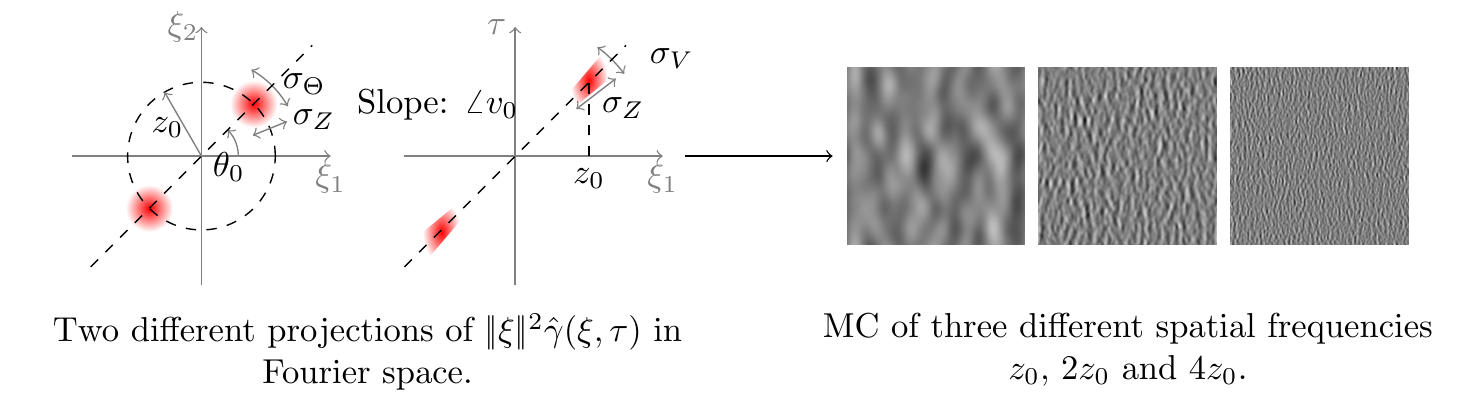}
\caption{Graphical representation of the covariance $\hat\ga$ shown as a projection on the spatial frequency plane (left) and the spatio-temporal frequency plane (middle) --- note the cone-like shape of the envelopes in both cases. The three luminance stimulus images on the right are an example of synthesized frames for three different spatial frequencies, respectively from left to right a low, a medium and a high frequency.}
\label{fig:rpt3D}
\end{center}
\end{figure}

\begin{figure}
\begin{center}
\includegraphics[scale=1]{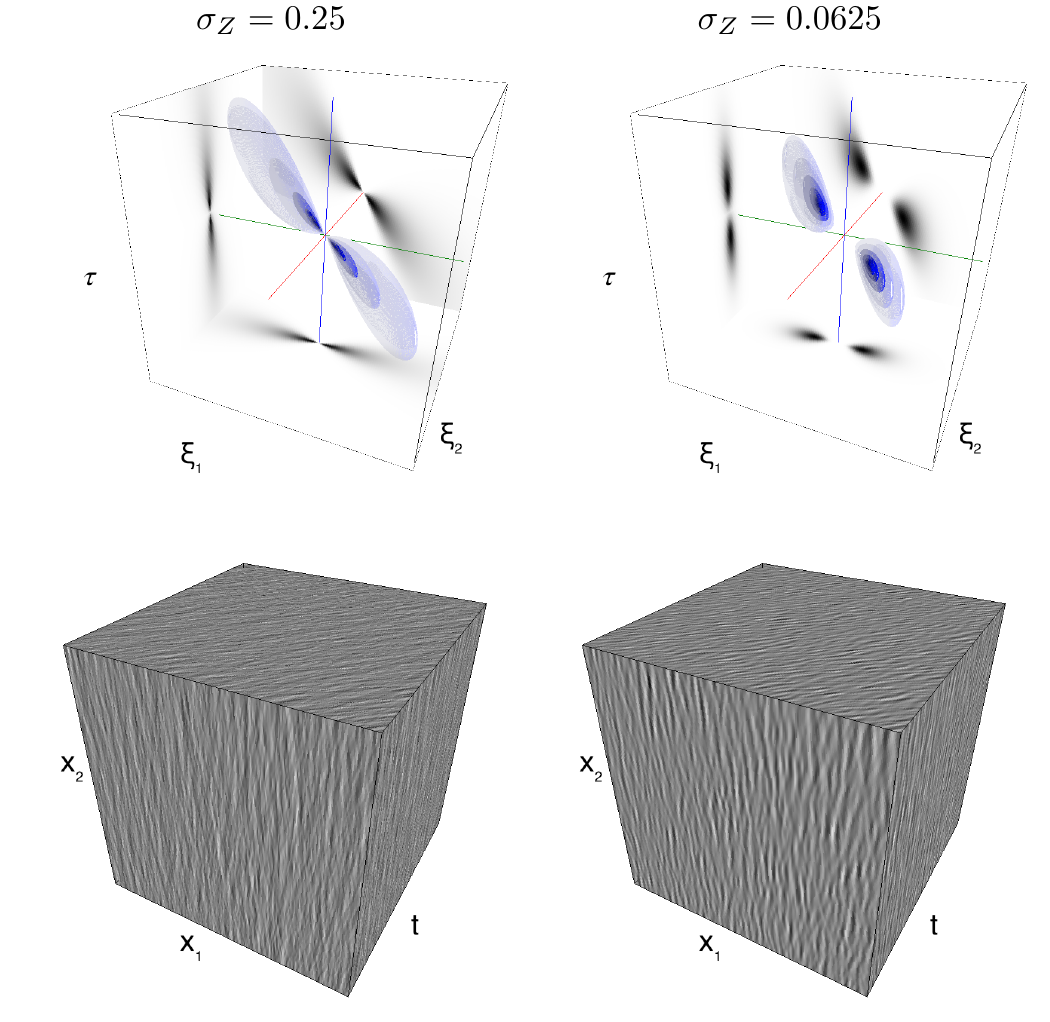}
\caption{Comparison of a broadband (left) vs. narrowband (right) stimulus. 
	Two instances (left and right columns) of two motions clouds having the same parameters, except the frequency bandwidths $\sz$, which were different.
    The top column displays iso-surfaces of $\hat\ga$, in the form of enclosing volumes at different energy values with respect to the peak amplitude of the Fourier spectrum.
    The bottom column shows an isometric view of the faces of a movie cube, which is a realization of the random field $I$.
    The first frame of the movie lies on the $(x_1,x_2,t=0)$ spatial plane.
	The Motion Cloud with the broadest bandwidth is often thought to best represent stereotyped natural stimuli since it similarly contains a broad range of frequency components.
}
\label{fig:wMC}
\end{center}
\end{figure}

Plugging these expressions~\eqref{eq-biosinspired-fz},~\eqref{eq-biosinspired-ftheta} and~\eqref{eq-biosinspired-fr} into the definition~\eqref{eq-dfn-mc-spectrum} of the power spectrum of the defintion of MCs, one obtains a parameterization which shares similarities with the one originally introduced in~\citep{Simoncini12}. 

Table~\ref{tab:model-param} which follows recaps the parameters of the biologically-inpired MC models. It is composed of the central parameters: $v_0$ for the speed, $\th_0$ for orientation  and $\z_0$ for the central spatial frequency modulus, as well as corresponding ``dispersion'' parameters $(\sr,\stheta,\lnbw)$ which account for the typical deviation around these central values. 
\begin{table}
\begin{center}
	\begin{tabular}{|c|c|c|c|}
		\hline
		Parameter name & Translation speed  & Orientation angle & Spatial freq. modulus \\\hline
		(mean, dispersion) &
			$(\vz, \sr)$ & 
			$(\th_0, \stheta)$ & 
			$(\z_0, \lnbw) $\\\hline
	\end{tabular}
\end{center}
\caption{The full set of six parameters which characterize the Motion Cloud stimulus model. See text for details.}
\label{tab:model-param}
\end{table}
Figure~\ref{fig:rpt3D} graphically shows the influence of these parameters on the shape of the MC power spectrum $\hat\ga$. 
 
We show in Figure~\ref{fig:wMC} two examples of such stimuli for two different spatial frequency bandwidths. This is particularly relevant as it is possible to dissociate the respective roles of broader or narrower spatial frequency bandwidths in action and perception~\citep{Simoncini12}. Using this formulation to extend the study of visual perception to other dimensions like orientation or speed bandwidths should provide a means to systematically titrate their respective role in motion integration and obtain a quantitative assessment of their respective contributions in experimental data.

%% file: NECO-07-17-2918-Source-sections/NECO-07-17-2918-Source-sec-mc-spde.tex
\section{sPDE Formulation and Synthesis Algorithm}
\label{sec-spde-model}

In this section, we show that the MC model (Definition~\ref{defn-mc-fourier}) can equally be described as the stationary solution of a stochastic partial differential equation (sPDE). This sPDE formulation is important since we aim to deal with dynamic stimulation, which should be described by a causal equation which is local in time. 
This is crucial for numerical simulations, since this allows us to perform real-time synthesis of stimuli using an auto-regressive time discretization. This is a significant departure from previous Fourier-based implementation of dynamic stimulations~\citep{Leon12,Simoncini12}. 
Moreover, this is also important to simplify the application of MC inside a Bayesian model of psychophysical experiments (see Section~\ref{sec-psychophysics}). 
In particular, the derivation of an equivalent sPDE model exploits a spectral formulation of MCs as Gaussian Random fields. The full proof along with the synthesis algorithm follows.

To be mathematically correct, all the sPDE in this article are written in the sense of generalized stochastic processes (GSP), which are to stochastic processes what generalized functions are to functions. This allows for the consideration of linear transformations of stochastic processes, like differentiation or Fourier transforms as for generalized functions. We refer to~\textcite{unser2014unified} for a recent use of GSP and to~\textcite{gelfand1964generalized} for the foundation of the theory. The connection between GSP and stochastic processes has been described in previous work~\citep{meidan1980connection}.


\subsection{Dynamic Textures as Solutions of sPDE}

\paragraph{Using a sPDE without global translation, $v_0=0$.}

We first give the definition of a sPDE cloud $I$ making use of another cloud $I_0$ without translation speed. This allows us to restrict our attention to the case $v_0=0$ in order to define a simple sPDE, and then to explicitly extend that result to the general case. 

\begin{defn}
For a given spatial covariance $\Si_W$, 2-D spatial filters $(\al,\be)$ and a translation speed $v_0 \in \RR^2$, a sPDE cloud is defined as
\eql{\label{eq-time-warping}
	I(x,t) \eqdef I_0(x-\vz t, t).
}
where $I_0$ is a stationary Gaussian field satisfying for all $(x,t)$,
\eql{\label{eq-spde}
	\Dd(I_0) = \pd{W}{t}
	\qwhereq
	\Dd(I_0) \eqdef \pdd{I_0}{t} + \al \star \pd{I_0}{t} + \be \star I_0
}
where the driving noise $\pd{W}{t}$ is white in time (i.e. corresponds to the temporal derivative of a Brownian motion in time) and has a 2-D stationary covariance $\sigma_W$ in space and $\star$ is the spatial convolution operator. 
\end{defn}

The random field $I_0$ solving equation~\eqref{eq-spde} thus corresponds to a sPDE cloud with no translation speed, $v_0=0$.
The filters $(\al,\be)$ parameterizing this sPDE cloud aim at enforcing an additional correlation of the model in time. 
Section~\ref{sec-spde} explains how to choose $(\al,\be,\sigma_W)$ so that these sPDE clouds, which are stationary solutions of equation~\eqref{eq-spde}, have the power spectrum given in~\eqref{eq-dfn-mc-spectrum} (in the case that $\vz=0$), i.e. are Motion Clouds.
Defining a causal equation which is local in time is crucial for numerical simulation (as explained in Section~\ref{sec-ar2-numerics}) but also to simplify the application of MC inside a Bayesian model of psychophysical experiments (see Section~\ref{sec:lkl-choice}).

The sPDE equation~\eqref{eq-spde} corresponds to a set of independent stochastic ODEs over the spatial Fourier domain, which reads, for each frequency $\xi$, 
\eql{\label{eq-spde-freq}
	\foralls t \in \RR, \quad
	\pdd{\hat I_0(\xi,t)}{t} + 
	\hat\al(\xi) \pd{\hat I_0(\xi,t)}{t} + 
	\hat\be(\xi) \hat I_0(\xi,t) = 
	\hsiW(\xi) \hat w(\xi,t)
}
where $\hat I_0(\xi,t)$ denotes the Fourier transform with respect to the spatial variable $x$ only and $\hsiW(\xi)^2$ is the spatial power spectrum of $\pd{W}{t}$, which means that 
\eql{\label{eq-cov-rhs-spde}
	\Si_W(x,y)=c(x-y)
	\qwhereq
	\hat c(\xi) = \hsiW^2(\xi).
}
Finally, $\hat w(\xi,t) \sim \Cc\Nn(0,1)$ where $\Cc\Nn$ is the complex-normal distribution. 

While the equation~\eqref{eq-spde-freq} should hold for all time $t \in \RR$, the construction of stationary solutions (hence sPDE clouds) of this equation is obtained by solving the sODE~\eqref{eq-spde-freq} forward for time $t>t_0$  with arbitrary boundary conditions at time $t=t_0$, and letting $t_0 \rightarrow -\infty$. This is consistent with the numerical scheme detailed in Section~\ref{sec-ar2-numerics}. 

The theoretical study of equation~\eqref{eq-spde} is beyond the scope of this paper, however one can show the existence and uniqueness of stationary solutions for this class of sPDE under stability conditions on the filters $(\al,\be)$ (see for instance~\citep{UnserBook,brockwell2009existence} and appendix Theorem~\ref{thm-sol-eq}). These conditions are automatically satisfied for the particular case of Section~\ref{sec-spde}.  

\paragraph{sPDE with global translation.}

The easiest way to define and synthesize a sPDE cloud $I$ with non-zero translation speed $v_0$ is to first define $I_0$ solving equation~\eqref{eq-spde-freq} and then translating it with constant speed using equation~\eqref{eq-time-warping}. An alternative way is to derive the sPDE satisfied by $I$, as detailed in the following proposition. This is useful to define motion energy in Section~\ref{sec:lkl-choice}. 

\begin{prop}\label{prop-spde-translation}
	The MCs noted $I$ with $(\al,\be,\Si_W)$ the speed parameters and $\vz$ the translation speed are the stationary solutions of the sPDE
	\eql{\label{eq-spde-warped}
		\Dd(I) + \dotp{\Gg(I)}{\vz} + \dotp{\Hh(I) \vz}{\vz} = \pd{W}{t}
	}
	where $\Dd$ is defined in equation~\eqref{eq-spde}, $\nabla_x^2 I$ is the Hessian of $I$ (second order spatial derivative) and where 
	\eql{\label{eq-defn-G-H}
		\Gg(I) \eqdef \al \star \nabla_x I + 2  \partial_t \nabla_x I 
		\qandq
		\Hh(I) \eqdef \nabla_x^2 I.
	}
\end{prop}
\begin{proof}
See Appendix~\ref{proof:prop-spde-translation}.
\end{proof}

\subsection{Equivalence between the spectral and sPDE formulations}
\label{sec-spde}

Since both MCs and sPDE clouds are obtained by a uniform translation with speed $v_0$ of a motion-less cloud, we can restrict our analysis to the case $v_0=0$ without loss of generality.

In order to relate MCs to sPDE clouds, equation~\eqref{eq-spde-freq} makes explicit that the functions $(\hat\al(\xi),\hat\be(\xi))$ should be chosen in order for the temporal covariance of the resulting process to be equal to (or at least to approximate well) the temporal covariance appearing in equation~\eqref{eq-dfn-mc-spectrum}.
This covariance should be localized around $0$ and be non-oscillating. It thus makes sense to constrain $(\hat\al(\xi),\hat\be(\xi))$ so that the corresponding ODE~\eqref{eq-spde-freq} be critically damped, which corresponds to imposing the following relationship
\eq{
	\foralls \xi, \quad
	\hat \al(\xi)=\frac{2}{\hat \nu(\xi)} \qandq \hat \be(\xi)=\frac{1}{\hat \nu^2(\xi)}
}
for some relaxation step size $\hat\nu(\xi)$.  The model is thus solely parameterized by the noise variance $\hsiW(\xi)$ and the characteristic time $\hat \nu(\xi)$.

The following proposition shows that the sPDE cloud model~\eqref{eq-spde} and the Motion Cloud model~\eqref{eq-dfn-mc-spectrum} are identical for an appropriate choice of function $\fr$. 

\begin{prop}\label{prop-spde-equiv}
	When considering 
	\eql{\label{eq-expression-fr} 
		\foralls r>0, \quad \fr(r) = \Ll^{-1}(h)(r/\sr) 
		\qwhereq
		h(u) = (1+u^2)^{-2} 
} 
where $\Ll$ is defined in equation~\eqref{eq-dfn-mc-spectrum}, 
	equation~\eqref{eq-spde} admits a solution $I$ which is a stationary Gaussian field with power spectrum defined in equation~\eqref{eq-dfn-mc-spectrum}, when setting
	\eql{\label{eq-equiv-spde-mc}
		\hsiW^2(\xi) =  \frac{4}{\hat\nu(\xi)^3 \norm{\xi}^2} \fz(\norm{\xi})
			\ftheta(\angle{\xi}), 
		\qandq
		\hat \nu(\xi) = 
		 \frac{1}{ \sr \norm{\xi} }.
	} 
\end{prop}
\begin{proof}
See Appendix~\ref{proof:prop-spde-equiv}.
\end{proof}

\paragraph{Expression for $\fr$}
\label{sec-expression-fr}

Equation~\eqref{eq-expression-fr} states that in order to obtain a perfect equivalence between the MC defined by equation~\eqref{eq-dfn-mc-spectrum} and by equation~\eqref{eq-spde}, the function $\Ll^{-1}(h)$ has to be well-defined. Therefore, we need to compute the inverse transform of the linear operator $\Ll$
\eq{\label{Ltransform}
	\forall u \in \RR,  \quad \Ll(f)(u) = 2\int_{0}^{\pi/2} f(-u/ \cos(\phi) ) \d \phi.  	
}
This is done for the function $h$ in Appendix~\ref{ap:comp-res}, Proposition~\ref{prop:inverse}.

\subsection{AR(2) Discretization of the sPDE}
\label{sec-ar2-numerics}

Most previous works on Gaussian texture synthesis (such as~\textcite{Galerne11} for static and~\textcite{Leon12,Simoncini12} for dynamic textures) have used a global Fourier-based approach and the explicit power spectrum expression~\eqref{eq-dfn-mc-spectrum}. The main drawbacks of such an approach are: (i)~it introduces an artificial periodicity in time and thus can only be used to synthesize a finite number of frames; (ii)~these frames must be synthesized at once, before the stimulation, which prevents real-time synthesis; (iii)~the discrete computational grid may introduce artifacts, in particular when one of the included frequencies is of the order of the discretization step or when a bandwidth is too small.

To address these issues, we follow the previous works of~\citep{Doretto-Dyntex,2014-xia-siims}
and make use of an auto-regressive (AR) discretization of the sPDE~\eqref{eq-spde}. In contrast with these previous works, we use a second order AR(2) regression (instead of a first order AR(1) model). Using higher order recursions is crucial to make the output consistent with the continuous formulation equation~\eqref{eq-spde}. Indeed, numerical simulations show that AR(1) iterations lead to unacceptable temporal artifacts: In particular, the time correlation of AR(1) random fields typically decays too fast in time. 

\paragraph{AR(2) synthesis without global translation, $v_0=0$.}

The discretization computes a (possibly infinite) discrete set of 2-D frames $(I_0^{(\ell)})_{\ell \geq \ell_0}$ separated by a time step $\De$, and we approach the derivatives at time $t=\ell\De$ as 
\eq{
	\pd{I_0(\cdot,t)}{t} \approx \De^{-1}( I_0^{(\ell)} - I_0^{(\ell-1)} )
	\qandq
	\pdd{I_0(\cdot,t)}{t} \approx \De^{-2}( I_0^{(\ell+1)} + I_0^{(\ell-1)} - 2 I_0^{(\ell)} ), 
}
which leads to the following explicit recursion 
\eql{\label{eq-ar-discr}
	\foralls \ell \geq \ell_0, \quad
	I_0^{(\ell+1)} = 
	( 2\de - \De \al - \De^2 \be) \star I_0^{(\ell)}  + 
	( -\de + \De \al ) \star I_0^{(\ell-1)} + 
	\De^2 W^{(\ell)}, 
}
where $\de$ is the 2-D Dirac distribution and where $(W^{(\ell)})_{\ell}$ are i.i.d. 2-D Gaussian field with distribution $\Nn(0,\Si_W)$, and $(I_0^{(\ell_0-1)}, I_0^{(\ell_0-1)})$ can be arbitrary initialized. 

One can show that when $\ell_0 \rightarrow -\infty$ (to allow for a long enough ``warmup'' phase to reach approximate time-stationarity) and $\De \rightarrow 0$, then $I_0^\De$ defined by interpolating $I_0^\De(\cdot,\De \ell) = I^{(\ell)}$ converges (in the sense of finite dimensional distributions) toward a solution $I_0$ of the sPDE~\eqref{eq-spde}. Here we choose to use the standard finite difference. However, we refer to~\textcite{Unser-Discrete,brockwell2007continuous} for more advanced discretization schemes.
We implement the recursion equation~\eqref{eq-ar-discr} by computing the 2-D convolutions with FFT's on a GPU, which allows us to generate high resolution videos in real time, without the need to explicitly store the synthesized video. 

\paragraph{AR(2) synthesis with global translation.}

The easiest way to approximate a sPDE cloud using an AR(2) recursion is to simply apply formula~\eqref{eq-time-warping} to $(I_0^{(\ell)})_\ell$ as defined in equation~\eqref{eq-ar-discr}, that is,  to define
\eq{
	I^{(\ell)}(x) \eqdef I_0^{(\ell)}(x-\vz \Delta \ell).
}
An alternative approach would consist in directly discretizing the sPDE~\eqref{eq-spde-warped}. We did not use this approach because it requires the discretization of spatial differential operators $\Gg$ and $\Hh$, and is hence less stable. 
A third, somehow hybrid, approach, is to apply the spatial translations to the AR(2) recursion, and define the following recursion 
\eql{\label{eq-ar-discr-translated}
	I^{(\ell+1)} = 
	\Uu_{v_0} \star I^{(\ell)}  + 
	\mathcal{V}_{v_0} \star I^{(\ell-1)} + 
	\De^2 W^{(\ell)},
}
\eql{\label{eq-defn-U-V}
	\qwhereq
	\choice{
		\Uu_{v_0} \eqdef (2\de - \De \al - \De^2 \be ) \star \de_{-\De v_0}, \\
		\mathcal{V}_{v_0} \eqdef (-\de + \De \al ) \star \de_{-2\De v_0}, 	
	}
}
where $\de_{s}$ indicates the Dirac at location $s$, so that $(\de_s \star I)(x) = I(x-s)$ implements the translation by $s$. Numerically, it is possible to implement equation~\eqref{eq-ar-discr-translated} over the Fourier domain, 
\eq{
	\hat I^{(\ell+1)}(\xi) =
	\hat\Uu_{v_0}(\xi) \hat  I^{(\ell)}(\xi)  +  
	\hat {\mathcal{V}}_{v_0}(\xi) \hat  I^{(\ell-1)}(\xi) + 
	\De^2 \hat \sigma_W(\xi) \hat w^{(\ell)}(\xi),
}
\eq{
	\qwhereq
	 \choice{
		\hat\Uu_{v_0}(\xi) \eqdef (2 - \De \hat \al(\xi) - \De^2 \hat\be(\xi)) e^{-\imath \De v_0 \xi}, \\
		\hat\Vv_{v_0}(\xi) \eqdef (-1 + \De \hat \al(\xi) ) e^{-2 \imath \De v_0 \xi}, 	
	}
}
and where $w^{(\ell)}$ is a 2-D white noise. 

%% file: NECO-07-17-2918-Source-sections/NECO-07-17-2918-Source-sec-psychophysics.tex
\section{An Empirical Study of Visual Speed Discrimination}
\label{sec-psychophysics}

To exploit the useful parametric transformation features of our Motion Clouds (MC) model and provide a generalizable proof of concept based on motion perception, we consider here the problem of judging the relative speed of moving dynamical textures. The overall aim is to characterize the impact of both average spatial frequency and average duration of temporal correlations on perceptual speed estimation, based on empirical evidence.

\subsection{Methods}
\label{sec-methods}
The task was to discriminate the speed $\q \in \RR$ of a MC stimulus moving with a horizontal central speed $\mathbf{\q} = (\q,0)$. We assign as the independent experimental variable the most represented spatial frequency $\lnz$, denoted $\nuis$ in the rest of the paper for easier reading. The other parameters are set to the following values 
\eq{
	\sr = \frac{1}{t^\star \z}, \quad
	\th_0 = \frac \pi 2, \quad
	\stheta = \frac{\pi}{12}.
 }
Note that $\sr$ is thus dependent on the value of $\z$ to ensure that $t^\star = \frac{1}{\sr\z}$ stays constant. This parameter $t^\star$ controls the temporal frequency bandwidth, as illustrated in the middle of Figure~\ref{fig:rpt3D}. 
We used a two alternative forced choice (2AFC) paradigm.
In each trial, a gray fixation screen with a small dark fixation spot was followed by two stimulus intervals of $250~\ms$ each, separated by an uniformly gray $250~\ms$ inter-stimulus interval. The first stimulus had parameters $(\q_1,\nuis_1)$ and the second had parameters $(\q_2,\nuis_2)$. 
At the end of the trial, a gray screen appeared asking the participant to report which one of the two intervals was perceived as moving faster by pressing one of two buttons, that is whether $\q_1>\q_2$ or $\q_2>\q_1$.

Given reference values $(\q^\star,\nuis^\star)$, for each trial, $(\q_1,\nuis_1)$ and $(\q_2,\nuis_2)$ are selected such that 
\eq{
	\choice{
		\q_i = \q^\star, \;
		\nuis_i \in \nuis^\star + \Delta_\Nuis \\
		\q_j \in \q^\star + \Delta_\Q, \;
		\nuis_j = \nuis^\star
	}
	\qwhereq
		\Delta_\Q = \{ -2, -1, 0, 1, 2 \},
  }
where $(i,j)=(1,2)$ or $(i,j)=(2,1)$ (i.e. the ordering is randomized across trials), and where $\nuis$ values are expressed in cycles per degree (\si{c/\degree}) and $\q$ values in \si{\degree/\second}. The range $\Delta_\Nuis$ is defined in Table~\ref{tab:exp}.
Ten repetitions of each of the 25 possible combinations of these parameters are made per block of 250 trials and at least four of such blocks were collected per condition tested. 
The outcome of these experiments are summarized by psychometric curve samples $\hat\phi_{\q^\star,\nuis^\star}$, where for all $(\q-\q^\star,\nuis-\nuis^\star) \in \Delta_\Q \times \Delta_\Nuis$, the value $\hat\phi_{\q^\star,\nuis^\star}(\q,\nuis)$ is modeled as a Bernoulli random variable with parameter $\phi_{\q^\star,\nuis^\star}(\q,\nuis)$ that a stimulus generated with parameters $(\q^\star, \nuis)$ is moving faster than a stimulus with parameters $(\q,\nuis^\star)$.

We tested different scenarios summarized in Table~\ref{tab:exp}. Each row corresponds to approximately 35 minutes of testing per participant and was always performed by at least three of the participants.
\begin{table} 
\begin{tabular}{|c|c|c|c|c|c|c|}
  \hline
   Case & $t^\star$ & $\lnsz$ & $\lnbw$ & $\q^\star$ & $\nuis^\star$ & $\Delta_\Nuis $ \\
  \thickhline
   A1 & $200~\ms$ & $1.0~\si{c/\degree}$ & $\times$ & $5~\si{\degree/\second}$ & $0.78~\si{c/\degree}$ & $\{ -0.31, -0.16, 0, 0.16, 0.47 \}$ \\
  \hline
   A2 & $200~\ms$ & $1.0~\si{c/\degree}$ & $\times$ & $5~\si{\degree/\second}$ & $1.25~\si{c/\degree}$ & $\{ -0.47, -0.31, 0, 0.31, 0.63 \}$ \\
  \thickhline
    A3 & $200~\ms$ & $\times$ & $1.28$ & $5~\si{\degree/\second}$ & $1.25~\si{c/\degree}$ & $\{ -0.47, -0.31, 0, 0.31, 0.63 \}$ \\
  \hline
    A4 & $100~\ms$ & $\times$ & $1.28$ & $5~\si{\degree/\second}$ & $1.25~\si{c/\degree}$ & $\{ -0.47, -0.31, 0, 0.31, 0.63 \}$ \\
  \hline
	A5 & $200~\ms$ & $\times$ & $1.28$ & $10~\si{\degree/\second}$ & $1.25~\si{c/\degree}$ & $\{ -0.47, -0.31, 0, 0.31, 0.63 \}$ \\
  \thickhline
\end{tabular}
\caption{Stimulus parameters for the range of tested experimental conditions. A1-A2 in the first two rows are both bandwidth controlled in \si{c/\degree} and A3-A5 are bandwidth controlled in octaves with high (A3 and A5) and low (A4) $t^\star$.}
\label{tab:exp}
\end{table}
Stimuli were generated using Matlab 7.10.0 on a Mac running OS 10.6.8 and displayed on a 20" Viewsonic p227f monitor with resolution $1024\times 768$ at 100~\si{\Hz}. Psychophysics routines were written using Matlab and Psychtoolbox 3.0.9 controlled the stimulus display. Observers sat 57~\si{\cm} from the screen in a dark room. Five male observers with normal or corrected-to-normal vision took part in these experiments. They gave their informed consent and the experiments received ethical approval from the Aix-Marseille Ethics Committee in accordance with the declaration of Helsinki.

\subsection{Psychometric Results}
\label{sec:psych-res}

Estimating speed in dynamic visual scenes is undoubtedly a crucial skill for successful interaction with the visual environment. Human judgements of perceived speed have therefore generated much interest, and been studied with a range of psychophysics paradigms. The different results obtained in these studies suggest that rather than computing a veridical estimate, the visual system generates speed judgements influenced by contrast~\citep{thompson1982perceived}, speed range~\citep{thompson2006speed}, luminance~\citep{hassan2015perceptual}, spatial frequency~\citep{brooks2011contrast,Simoncini12,Smith10} and retinal eccentricity~\citep{hassan2016perceived}. There are currently no theoretical models of the underlying mechanisms serving speed estimation which capture this dependence on such a broad range of image characteristics. One of the reasons for this might be that the simplified grating stimuli used in most of the previous studies does not allow experimenters to shed light on the possible elaborations in neural processing that arise when more complex natural or naturalistic stimulation is used. Such elaborations, like nonlinearities in spatio-temporal frequency space, can be seen in their simplest form even with a superposition of a pair of gratings~\citep{Priebe03}. In the current work, we use our formulation of motion cloud stimuli, which allows for separate parametric manipulation of peak spatial frequency ($z$), spatial frequency bandwidth ($B_z, \sigma_z$) and stimulus lifetime ($t^\star$), which is inversely related to the temporal variability. The stimuli are all broadband, closely resembling the frequency properties under natural stimulation. Our approach is to test five participants under several parametric conditions given in Table~\ref{tab:exp} and using a large number of trials.

\paragraph{Psychometric function estimation} The psychometric function is estimated by the the following sigmoidal template function
\eql{\label{eq:psych-func}
	\phi_{\q^\star,\nuis^\star}^{{\mu},\Sigma}( \q,\nuis ) 
	= 
	\psi\pa{ 
		\frac{ 
			\q-\q^\star - {\mu}_{\nuis,\nuis^\star} 
		}{ 
			 \Sigma_{\nuis,\nuis^\star}   } 
		}	 	
	}
where $\psi(t)=\frac{1}{\sqrt{2\pi}}\int_{-\infty}^t e^{-s^2/2} \d s$ is the cumulative normal function and $({\mu}_{\nuis,\nuis^\star},\Sigma_{\nuis,\nuis^\star})$ denotes respectively bias and inverse sensitivity. The collected data are used to fit the two parameters using maximum likelihood estimation (see~\citep{wichmann2001psychometric})
\eq{
(\hat \mu,\hat \Sigma) = \uargmin{\mu,\Sigma} \sum_\q \KL{\hat \phi_{\q^\star,\nuis^\star}}{\phi_{\q^\star,\nuis^\star}^{\mu,\Sigma}}
}
where $\KL{\hat p}{p}$ is the Kullback-Leibler divergence between samples $\hat p$ and model $p$ under a Bernouilli distribution
\eq{
\KL{\hat p}{p} = \hat p \log\pa{\frac{\hat p}{p}} + (1-\hat p) \log\pa{\frac{1-\hat p}{1-p}}.
}
Results of these estimations are shown in Figure~\ref{fig:psych} for both non-parametric and linear $\Sigma_{\nuis,\nuis^\star}$ fits. 
\begin{rem} In practice we perform the fit in the log-speed domain, \textit{ie} we consider $\phi_{\tilde \q^\star,\nuis^\star}( \tilde\q,\nuis )$ where $\tilde \q = \ln(1 + \q / \q_0) $ with $\q_0 = 0.3 \si{\degree/\second} $ following \citep{Stocker06}. As the estimated bias $\tilde{\mu}$ is obtained in the log-speed domain, we convert it back to the speed domain by computing ${\mu}$, which solves the following equation 
\eq{
 \log(1+(v^\star+{\mu})/v_0) = \log(1+v^\star/v_0) + \tilde{\mu}
} 
Then, the speed bias is $ {\mu} = (v_0+v^\star)(\exp(\tilde{\mu})-1) $.
\end{rem}
\begin{figure}[!ht]
\begin{center}
\makebox[\textwidth]{\includegraphics[scale=0.36]{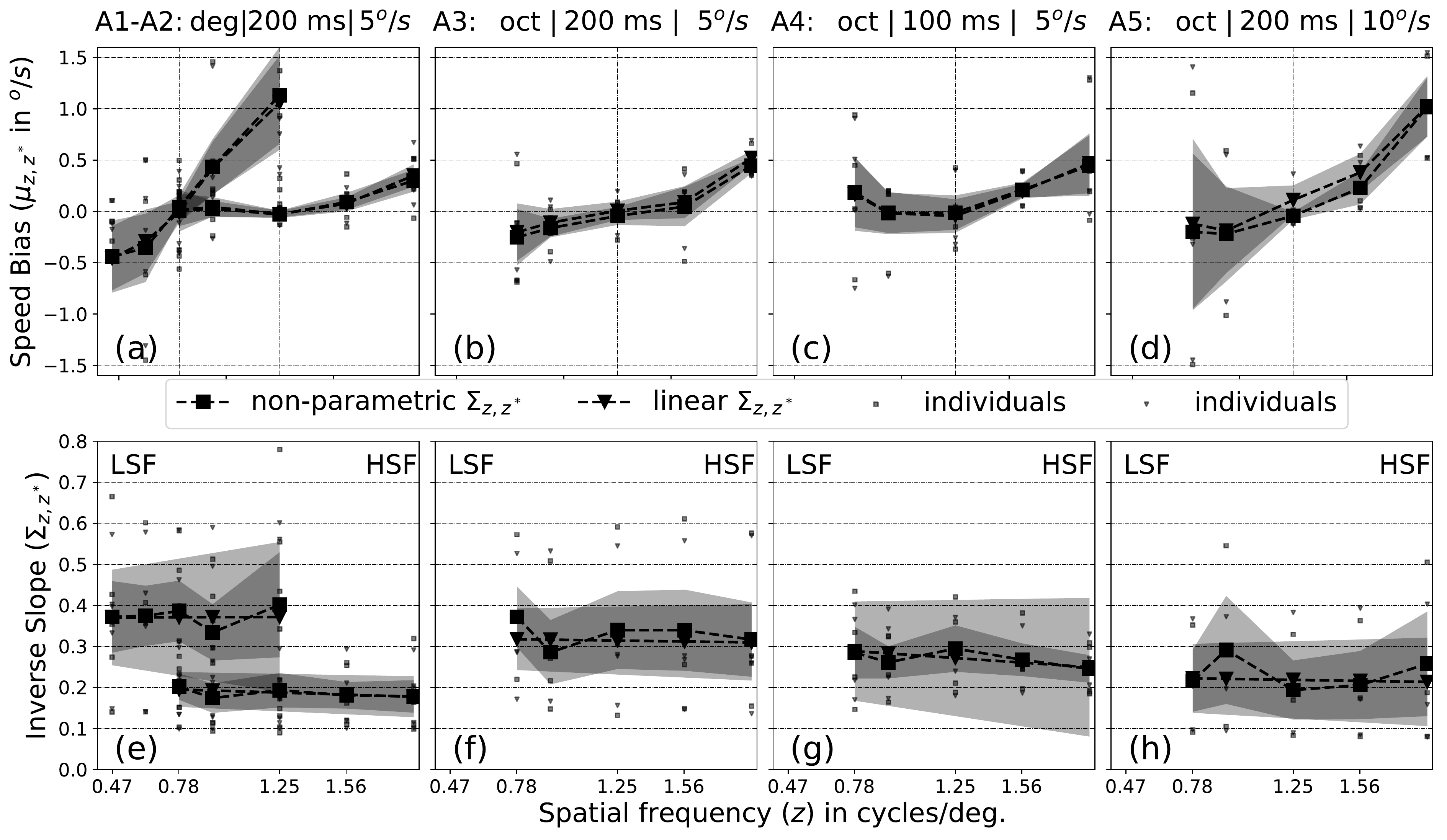}}
\caption{Averaged results over participants for the perceptual biases (top row) and inverse slope (bottom row) plotted against the tested central spatial frequency $z$. The specific parameters for each column are indicated above: bandwidth in degree (deg) or octaves (oct.), value of stimulus lifetime $t^\star$ and reference speed $v^\star$.  Small markers represent individual results and large markers represent population average. \textbf{(a-d)} \textit{Speed biases which generally show an increase at higher frequencies, but with individual differences.} \textbf{(e-h)} \textit{Inverse psychometric slopes that generally appear to be constant or decreasing across frequency.} From left to right: conditions A1-A2, A3, A4 and A5 (see Table 2 for details).}
\label{fig:psych-simp}
\end{center}
\end{figure}
%
%
\paragraph{Cycle-controlled bandwidth conditions}
The main manipulation in each case is the direct comparison of the speed of a range of five stimuli in which the central spatial frequency varies between five values, but all other parameters are equated under the different conditions. In a first manipulation in which bandwidth is controlled by fixing it at a value of 1~\si{c/\degree} for all stimuli (conditions A1-A2 in Table~\ref{tab:exp}), we find that lower frequencies are consistently perceived to be moving slower than higher frequencies (see Figure~\ref{fig:psych-simp}a). This trend is the same for both the lower and the higher spatial frequency ranges used in the tasks, yet the biases are smaller for the higher frequency range (see A1-A2 in Table~\ref{tab:exp} for details). This suggests that the effect generalizes across the two scales used, but that shifting the central spatial frequency value of the stimulus which forms the reference scale results in a change in sensitivity during speed discrimination. For example, comparing A1 and A2 performance in Figure~\ref{fig:psych-simp}, when the five stimuli of different speed which make up the reference scale are changed from $z^\star=0.78$ (A1) to $z^\star=1.25$ (A2), speed estimates seem to become less reliable. The same comparison is using a different psychometric measurement scale in each case.  The sensitivity to the discrimination of stimuli measured in the inverse of the psychometric slope is found to remain approximately constant across the range of frequency tested for each of the tested spatial frequencies, see Figure~\ref{fig:psych-simp}e. However, the sensitivity increases significantly ($\Sigma_{z^\star,z}$ decreases) from condition A1 to condition A2. Such an effect suggests that an increasing trend in sensitivity may exist (See paragraph~\ref{para:sensitive}). 

%
\paragraph{Octave-controlled bandwidth conditions}
The octave-bandwidth controlled stimuli of conditions A3 to A5 (see Table~\ref{tab:exp}), allow us to vary the spatial frequency manipulations ($z$) in a way that generates scale invariant bandwidths exactly as would be expected from zooming movements towards or away from scene objects (see Figure~\ref{fig:warps}). Thus if trends seen in Figure~\ref{fig:psych-simp}e were purely the result of ecologically invalid fixing of bandwidths at 1~\si{c/\degree} in the manipulations, this would be corrected in the current manipulation. Only the higher frequency comparison range from conditions A2 is used because trends are seen to be consistent across conditions A1 and A2. We find that the trends are generally the same as that seen in Figure~\ref{fig:psych-simp}a. Indeed higher spatial frequencies are consistently perceived as faster than lower ones, as shown in Figure~\ref{fig:psych-simp}b-d. Interestingly, for the degree bandwidth controlled stimuli, the biases are lower than those for the equivalent octave controlled stimuli (e.g. compare Figure~\ref{fig:psych-simp}a with~\ref{fig:psych-simp}b). This can also be seen in Figure~\ref{fig:biases-ampli} (conditions A2 and A3). A change in the bias is also seen with the manipulation of $t^\star$, as increasing temporal frequency variability when going from biases in Figure~\ref{fig:psych-simp}b to those in Figure~\ref{fig:psych-simp}c entails a reduction in measured biases, with an effect of about 25 \% which is also visible in Figure~\ref{fig:biases-ampli} (conditions A3 and A4 for M1).   

\paragraph{Is sensitivity dependent on stimulus spatial frequency ? \label{para:sensitive}} To explore further the sensitivity trend, we fit the data with a psychometric function by assuming a linear model for $\Sigma_{\nuis,\nuis^\star}$ and test for a significant negative slope. None of the slopes are significantly different from $0$ at the population level. At the individuals' level, among all conditions and subjects, we find that $13$ out of $21$ slopes were significantly decreasing. Therefore, we interpret this as a possible decrease in sensitivity at higher $z$ seen in $13$ out of $21$ of the cases, but one which shows large individual differences in sensitivity trends.

\paragraph{Qualitative results summary}
\begin{enumerate}[label=(\alph*)]
\item \label{qual1} spatial frequency has a positive effect on perceived speed ($\mu_{\nuis,\nuis^\star}$ increases as $z$ increases),
\item \label{qual2} the inverse sensitivity remains constant or is decreasing with spatial frequency (resp. $ \Sigma_{\nuis,\nuis^\star} $ does not depend on $z$ or decreases as $z$ increases) but there are large individual differences in this sensitivity change.
\end{enumerate}
In the next section, we detail a Bayesian observer model to account for these observed effects. 

\subsection{Observer Model}
\label{sec:obs-model}

We list here the general assumptions underlying our model:
\begin{enumerate}[label=(\roman*)]
\item \label{assump1} the observer performs abstract measurement of the stimulus denoted by a real random variable $M$,
\item \label{assump2} the observer estimates speed using an estimator based on the posterior speed distribution $ \PP_{\Q|\M} $,
\item \label{assump3} the posterior distribution is implicit, Bayes rule states that $ \PP_{\Q|\M} \propto \PP_{\M|\Q}\PP_{\Q} $,
\item \label{assump4} the observer knows all other stimulus parameters (in particular the spatial frequency $\nuis$),
\item \label{assump5} the observer takes a decision without noise.
\end{enumerate}
These asssumptions corresponds to the ideal Bayesian observer model. We detail below the relation between this model and the psychometric bias and inverse sensitivity $(\mu_{\nuis,\nuis^\star},\Sigma_{\nuis,\nuis^\star})$. We also give details to derive the likelihood directly from the MC model and discuss the expected consequences.

\subsubsection{Ideal Bayesian Observer}
\label{sec:idealObs}
The assumptions~\ref{assump1}-\ref{assump5} correspond to the methodology of the Bayesian observer used for instance in~\textcite{Stocker06,SotiropoulosVR,jogan2015signal}. This previous work provides the foundation for the work on Bayesian Observer models in perception on which we build our modifications accounting for our naturalistic dynamic stimulus case. We assume that the posterior speed distribution may depend on spatial frequency because any observed effects must come from the change in spatial frequency and the effect it may have on the likelihood. This assumption is also motivated by a body of empirical evidence showing consistent effects of spatial frequency changes on speed estimation~\citep{brooks2011contrast,vacher2015biologically}. Findings from primate neurophysiology probing extra-striate cortical neurons with compound gratings also show that speed is estimated by neural units whose speed response (i.e. not just response variance associated with likelihood widths) is highly dependent on spatio-temporal frequency structure~\citep{Priebe03,Perrone01}.  Finally, we also assume that the observer measures speed using a Maximum A Posteriori (MAP) estimator
\eql{\label{eq-map}
\begin{split}
	\hat \q(\m) & = \uargmax{\q} \PP_{\Q|\M,\Nuis}(\q|\m,\nuis) \\
						& = \uargmin{\q} [ -\log(\PP_{\M|\Q,\Nuis}(\m|\q,\nuis)) - \log(\PP_{\Q|\Nuis}(\q|\nuis)) ]
\end{split}	
	}
computed from the internal representation $\m \in \RR$ of the observed stimulus. Note that the distribution of measurements (the likelihood) and the prior are both conditioned on spatial frequency $\nuis$. As the likelihood is also obviously conditioned on speed, we denote measurement as $M_{\q,\nuis}$. To simplify the numerical analysis, we assume a Gaussian likelihood (in log-speed domain), with a variance independent of $\q$ consistently with previous literature~\citep{Stocker06,SotiropoulosVR,jogan2015signal}. Furthermore, we assume that the prior is Laplacian (in log-speed domain) as this gives a good description of the a priori statistics of speeds in natural images~\citep{Dong10}: \eql{\label{eq-likelihood-prior}
	\PP_{\M|\Q,\Nuis}(\m|\q,\nuis) = \frac{1}{\sqrt{2\pi} \sigma_\nuis} e^{ -\frac{|\m-\q|^2}{2\sigma_\nuis^2} }
	\qandq
	\PP_{\Q|\Nuis}(\q|\nuis)=\PP_\Q(\q) \propto e^{a \q}
}
where $a<0$. 
\begin{rem}
\label{rem:az}
We initially assume that the posterior speed distribution is conditioned on spatial frequency, thus the likelihood and prior distributions also depend on spatial frequency. However, there is currently no conclusive support in favor of a spatial frequency dependent speed prior in the literature, but evidence of spatial frequency influencing speed estimation is discussed in the previous paragraph. Therefore, only the likelihood width $\sigma_z$ depends on spatial frequency $z$ and the log-prior slope $a$ does not. We discuss in more details the choice of the likelihood and its dependence on spatial frequency in Section~\ref{sec:lkl-choice}.
\end{rem} 
Figure~\ref{fig:inference} shows an example of how the likelihood and prior described in equation~\eqref{eq-likelihood-prior} combine into a posterior distribution that resembles a shifted version of the likelihood. In practice, we are able to compute the distribution of the estimates $ \hat \q(\M_{\q,\nuis}) $ as stated in the following proposition.
\begin{prop}
\label{prop:distrib-estim}
In the special case of the MAP estimator~\eqref{eq-map} with a parameterization defined in equation~\eqref{eq-likelihood-prior}, one has
\eql{\label{eq:distrib-estim}
	\hat\q(\M_{\q,\nuis}) \sim \Nn(\q+a \sigma^2_\nuis, \sigma^2_\nuis).
}
\end{prop}

\begin{figure}[!ht]
\begin{center}
\includegraphics[scale=1]{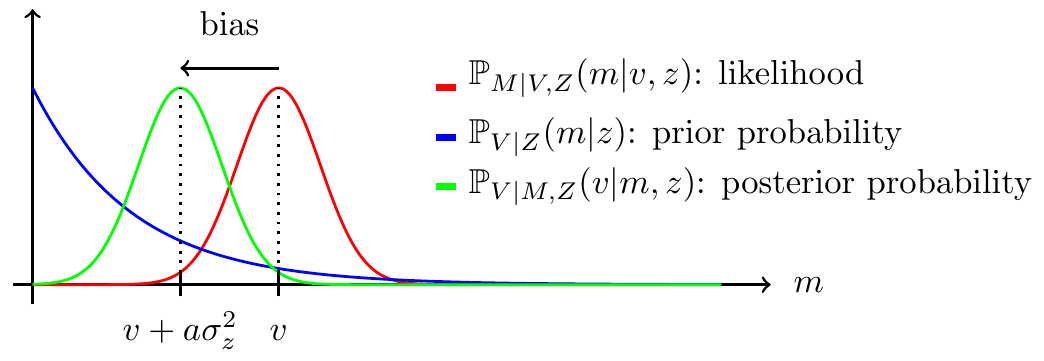}
\caption{Multiplying a Gaussian likelihood by a Laplacian prior gives a Gaussian posterior that is similar to a shifted version of the likelihood.}
\label{fig:inference}
\end{center}
\end{figure}

Once the observer has estimated the speed of two presented stimuli, he must take a decision to judge which stimulus was faster. Following assumption~\ref{assump5}, the decision is ideal in the sense that it is performed without noise. In other words, the observer compares the two speeds and decides whether $(\hat\q(\m_{\q,\nuis^\star}),\hat\q(\m_{\q^\star,\nuis}))$ belongs to the decision set $E = \{(\q_1,\q_2) \in \RR^2 | \q_1>\q_2 \}$ or not. Thus, we define the theoretical psychometric curve of an ideal Bayesian observer as 
\eq{
	\phi_{\q^\star,\nuis^\star}( \q,\nuis ) \eqdef \EE( \hat\q(\M_{\q,\nuis^\star}) > \hat\q(\M_{\q^\star,\nuis}) )
} 
Following proposition~\ref{prop:distrib-estim}, in our special case of Gaussian likelihood and Laplacian prior, the psychometric curve can be computed in closed form. 

\begin{prop}
\label{prop:psycurve}
In the special case of the MAP estimator~\eqref{eq-map} with a parameterization defined in equation~\eqref{eq-likelihood-prior}, one has
\eql{\label{eq-psycurve-theo}
	\phi_{\q^\star,\nuis^\star}( \q,\nuis ) 
	=
	\phi_{\q^\star,\nuis^\star}^{a,\sigma}( \q,\nuis ) 
	\eqdef
	\psi\pa{ 
		\frac{ 
			\q-\q^\star + a (\sigma_{\nuis^\star}^2  - \sigma_{\nuis}^2)
		}{ 
			\sqrt{ \sigma_{\nuis^\star}^2 + \sigma_{\nuis}^2  } 
		}
	}
}
where $\psi$ is defined in equation~\eqref{eq:psych-func}.
\end{prop}
\begin{proof}
See Appendix~\ref{proof:prop-psycurve}.
\end{proof}
Proposition~\ref{prop:psycurve} provides the connection between the Bayesian model parameters and the classical psychometric measures of bias and sensitivity. In particular, it explains the heuristic sigmoidal templates commonly used in psychophysics, see Section~\ref{sec:psych-res}. An example of two psychometric curves is shown in Figure~\ref{fig:psych_curve}. We have the following relations:

\noindent\begin{minipage}{0.5\linewidth}
\eql{\label{eq:rel-bias-bayes}
\mu_{\nuis,\nuis^\star} =  a(\sigma_{\nuis}^2 - \sigma_{\nuis^\star}^2) 
}
\end{minipage}%
\begin{minipage}{0.5\linewidth}
\eql{\label{eq:rel-sensitive-bayes}
\text{and} \quad \quad \Sigma_{\nuis,\nuis^\star}^2 = \sigma_{\nuis^\star}^2 + \sigma_{\nuis}^2.
}
\end{minipage}

\begin{rem}
The experiment allows us to estimate bias and inverse sensitivity $(\mu_{\nuis,\nuis^\star},\Sigma_{\nuis,\nuis^\star})$. Knowing these parameters, it is possible to recover parameters of the ideal Bayesian observer model. Equation~\ref{eq:rel-sensitive-bayes} has a unique solution and equation~\ref{eq:rel-bias-bayes} can be solved using the least square estimator. 
\end{rem}
\begin{rem}
Under this model, a positive bias comes from a decrease in likelihood width and a negative log-prior slope. As concluded in Section~\ref{sec:psych-res}, we observe a significant decrease in inverse sensitivity in $13$ out of $21$ subjects and conditions. Therefore, the model, when fitted to the data, will force the likelihood width to decrease. Further experiments will be necessary to verify the significance of this observation. Yet, the model is well supported by the literature (see~\cite{Stocker06,jogan2015signal,SotiropoulosVR}) and is compatible with the properties of the stimuli (see Section~\ref{sec:lkl-choice}).
\end{rem}
\begin{figure}[ht]
\begin{center}
\includegraphics[scale=1]{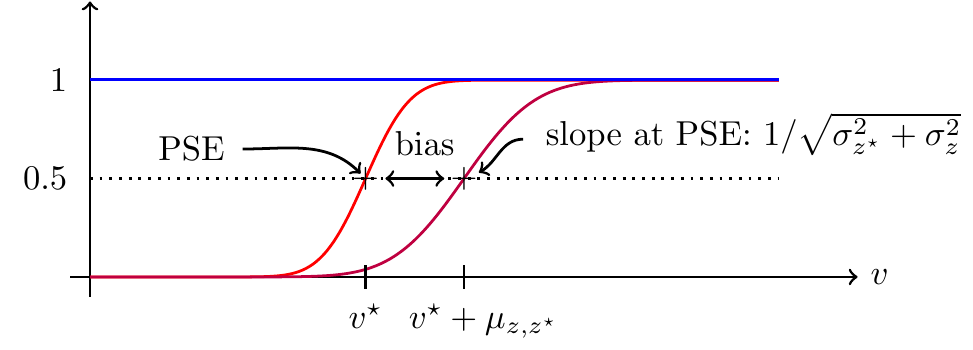}
\caption{The shape of the psychometric function follows the estimation of the two speeds by the Bayesian inference described in Figure~\ref{fig:inference}. This figure illustrates Proposition~\ref{prop:psycurve}. The bias ensues from the difference between the bias on the two estimated speeds. } 
\label{fig:psych_curve}
\end{center}
\end{figure}
%

\subsubsection{Discussion: Likelihood}
\label{sec:lkl-choice}

A MC $I_{\q,\nuis}$  is a random Gaussian field of power spectrum defined by equation~\eqref{eq-dfn-mc-spectrum}, with central speeds $\vz=(\q,0)$ and central spatial frequency $\nuis$ (the other parameters being fixed, as explained in Section~\ref{sec-methods}). Assuming that the abstract measurements correspond to the presented frames \textit{ie} $ \M_{\q,\nuis} = I_{\q,\nuis} $ it is possible to use the MC generative model as a likelihood. In the absence of a prior, the MAP estimator is equal to the Maximum Likelihood Estimator (MLE)
\eql{\label{eq-mle-mc}
	\hat \q^{}(m) = \hat \q^{\text{MLE}}(i) = \uargmin{\q}  -\log(\PP_{I|\Q,\Nuis}(i|\q,\nuis)). 
}

Thanks to the sPDE formulation, it is possible to give a simple rigorous expression for $ -\log(\PP_{I|\Q,\Nuis}(i|\q,\nuis)) $ in the case of discretized clouds satisfying the AR(2) recursion equation~\eqref{eq-ar-discr-translated}.  In this case, for some input video $I_{\q,\nuis}=(I^{(\ell)})_{\ell=1}^L$, the log-likelihood reads 
\eq{
	-\log(\PP_{I|\Q,\Nuis}(I_{\q,\nuis}|\q,\nuis)) = \tilde Z_I + K_{v_0}(I_{\q,\nuis})
	\qwhereq
}
\eq{
	K_{v_0}(I_{\q,\nuis}) \eqdef
	\frac{1}{\De^4} \sum_{\ell=1}^L \int_\Om 
	|
		K_W \star I^{(\ell+1)}(x) -
		\Uu_{v_0} \star K_W \star I^{(\ell)}(x)  -  
		\mathcal{V}_{v_0} \star K_W \star   I^{(\ell-1)}(x) 
	|^2 
	\d x	
}
where $\Uu_{v_0}$ and $\mathcal{V}_{v_0}$ are defined in equation~\eqref{eq-defn-U-V} and where $K_W$ is the spatial filter corresponding to the square-root inverse of the covariance $\Si_W$, that is, which satisfies $\hat K_W(\xi) \eqdef \hsiW(\xi)^{-1}$.
This convenient formulation can be used to re-write the MLE estimator of the horizontal speed $\q$ parameter of a MC as 
\eql{\label{eq-optim-mle}
	\hat \q^{\text{MLE}}(i) = \uargmin{\q} K_{v_0}(i)
	\qwhereq
	\vz = (\q,0) \in \RR^2
}
where we used the fact that $\tilde Z_I$ is independent from $\vz$.

The solution to this optimization problem with respect to $v$ is computed using the Newton-CG optimization method implemented in the Python library \texttt{scipy}. In Figure~\ref{fig:mle}a, we show a histogram of speed estimates $\hat \q^{\text{MLE}}(I_{\q,\nuis})$ performed over 200 Motion Clouds generated with speed $\q = 6~\si{\degree/\second}$ and spatial frequency $\nuis = 0.78~\si{c/\degree}$. In Figure~\ref{fig:mle}b, we show the evolution of the standard deviation of speed estimates $\hat \q^{\text{MLE}}(I_{\q,\nuis})$ as a function of spatial frequencies $\nuis \in \{0.47~\si{c/\degree}, 0.62~\si{c/\degree}, 0.78~\si{c/\degree}, 0.94~\si{c/\degree}, 1.28~\si{c/\degree}\}$. For each spatial frequency, estimates are again similarly obtained over a set of 200 Motion Clouds generated with speed $\q = 6~\si{\degree/\second}$. First, we observed that $\hat \q^{\text{MLE}}(I_{\q,\nuis})$ is well approximated by a Gaussian random variable with mean $\q$. Second, the standard deviation of these estimates decreases when the spatial frequency increases. The two conclusions follow the fact that our model is Gaussian and that we impose the relation $ \sr = 1/(t^\star \z) $ \textit{ie} that standard deviation of speed is inversely proportional to spatial frequency. The decreasing trend combined with a prior for slow speed $ a < 0$ would reproduce the positive bias of spatial frequency over speed perception observed in Section~\ref{sec:psych-res}.
\begin{figure}[!ht]
\begin{center}
\includegraphics[scale=1]{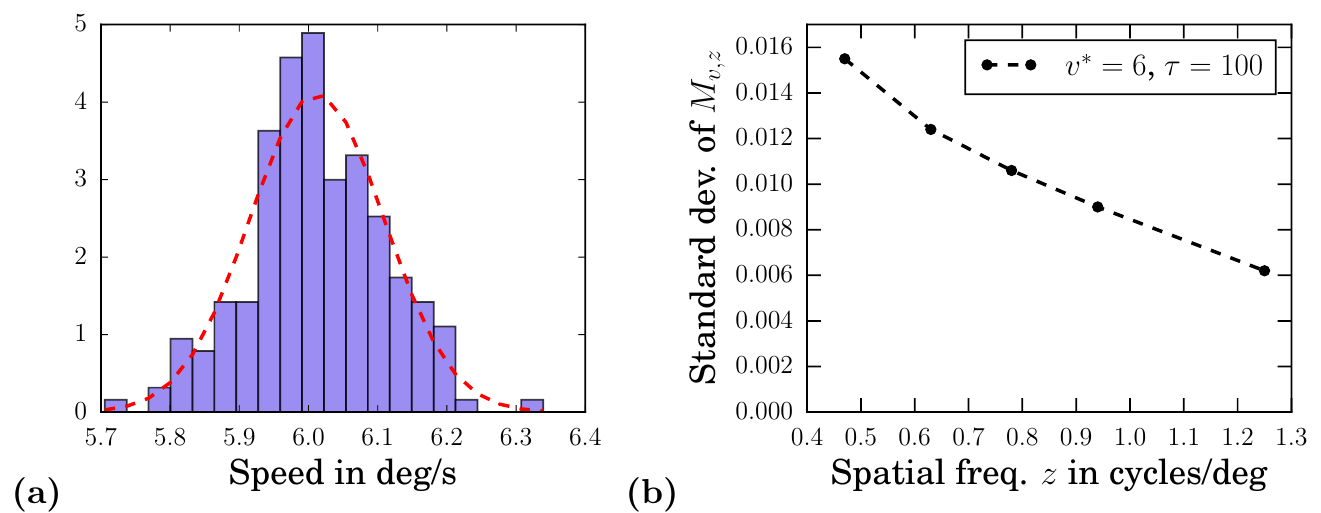}
\caption{Simulation of the speed distributions of a set of motion clouds with the experimentally tested parameters. \textbf{(a)} Histogram of estimates of $\hat \q^{\text{MLE}}(I_{\q,\nuis})$ for $z=0.8~\si{c/\degree}$ defined by equation~\eqref{eq-mle-mc}. These estimates are well approximated by a Gaussian distribution (red dotted line). \textbf{(b)} Standard deviation of estimates of $\hat \q^{\text{MLE}}(I_{\q,\nuis})$ as a function of~$z$. The standard deviation of these estimates is inversely proportional to the spatial frequency $z$.}
\label{fig:mle}
\end{center}
\end{figure}
If a human subject were estimating speed in such an optimal way, equation~\eqref{eq:rel-sensitive-bayes} indicates that inverse sensitivity $\Sigma_{\nuis,\nuis^\star}$ would also be inversely proportional to spatial frequency. Yet, the primary analysis conducted in Section~\ref{sec:psych-res} do not give a clear trend for the inverse sensitivity. As a consequence, the analysis conducted in Section~\ref{sec:psych-res} is ambiguous and does not allow us to make definitive conclusions about the compatibility of the MC model and the existing literature with the observed human performances. 

%

\subsection{Likelihood and Prior Estimation}
\label{sec:likeli-prior}
In order to fit this model to our data we use an iterative two-step method, each consisting in minimizing the Kullback-Leibler divergence between the model and its samples. This process is the equivalent of a maximum likelihood estimate. The first step consists in fitting each psychometric curve individually, then, for the second step, to use the results as a starting point to fit all the psychometric curves together. Numerically,  we used the Nelder-Mead simplex method as implemented in the python library \texttt{scipy}. 
\begin{itemize}[leftmargin=18mm]
\label{item:fit-procedure}
\item[Step 1:] for all $z,z^\star$, initialized at a random point, compute 
\eq{
(\bar \mu,\bar \Sigma) = \uargmin{\mu,\Sigma} \sum_\q \KL{\hat \phi_{\q^\star,\nuis^\star}}{\phi_{\q^\star,\nuis^\star}^{\mu,\Sigma}}
}
where $\phi_{\q^\star,\nuis^\star}^{\mu,\Sigma}$ is defined in equation~\eqref{eq:psych-func}.
\item[Step 2:]  solve the equations~\eqref{eq:rel-bias-bayes} and~\eqref{eq:rel-sensitive-bayes} between $(\bar \mu,\bar \Sigma)$ and $(\bar a, \bar \sigma)$, initialize at $(\bar a, \bar \sigma)$ and compute 
\eq{
(\hat a , \hat \sigma) = \uargmin{a,\sigma} \sum_{\nuis,\nuis^\star} \sum_\q \KL{\hat \phi_{\q^\star,\nuis^\star}}{\phi_{\q^\star,\nuis^\star}^{a,\sigma}}
}
where $\phi_{\q^\star,\nuis^\star}^{a,\sigma}$ is defined in equation~\eqref{eq-psycurve-theo}.
\end{itemize}

We use a repeated stochastic initialization in the first step in order to overcome the presence of local minima encountered during the fitting process. The approach was found to exhibit better results than a direct and global fit (third point). 



\subsection{Modeling results}
\label{sec-results}

We use the Bayesian formulation detailed in Section~\ref{sec:idealObs} and the fitting process described in Section~\ref{sec:likeli-prior} to estimate, for each subject, the likelihood widths and the corresponding log-prior slopes under the tested experimental conditions. We plot in Figure~\ref{fig:psych} the fit of bias and inverse sensitivity for the sigmoid model~\ref{sec:psych-res} and Bayesian model~\ref{sec:idealObs} averaged over subjects. Figure~\ref{fig:model} displays the corresponding likelihood widths and log-prior slopes for the Bayesian model also averaged over subjects. Error bars correspond to the standard deviation of the mean. 

\paragraph{Measured biases and inverse sensitivity}
As shown in Figure~\ref{fig:psych}, both models M2 and M3 correctly account for the biases and inverse sensitivity estimated with model M1 (see Section~\ref{sec:psych-res}) except for conditions A1 and A2. For condition A1 (Figure~\ref{fig:psych}a), the bias is underestimated by model M2 and M3 compared to model M1. For condition A2 (Figure~\ref{fig:psych}e), the inverse sensitivity is overestimated by model M2 and M3 compared to model M1. The observed differences come from the fact that in models M2 and M3 the overlapping spatial frequencies of conditions A1 and A2 are pulled together. As a consequence, the fit is more constrained than for model M1 and is therefore smoother. While that discrepancy does not affect our conclusion, it raises the question of pulling different overlapping conditions together. The overlapping tested spatial frequencies are together whereas they were collected with different reference spatial frequencies such that the sensitivity of each of the the psychometric speed measurement scales appears to have been different. Despite averaging over subjects, the Bayesian estimates of inverse sensitivity appear smoother than the sigmoid estimates (see Figure~\ref{fig:psych}f and~\ref{fig:psych}h). Finally, a clearer decreasing trend is visible in the Bayesian estimates of inverse sensitivity. 

\paragraph{Corresponding sensory likelihood widths}
There is a systematic decreasing trend within the likelihood width fits in Figure~\ref{fig:model}a-d excepting Figure~\ref{fig:model}c which shows an inverted U-shape. The fact that all subjects did not run all experimental conditions explains this difference ($2$ subjects out of $4$ show a U-shape bias, see Figure~\ref{fig:psych}c). Subject to subject variability is similar for all conditions except for the least temporal variability for which it is smaller (Figure~\ref{fig:model}c). 

\paragraph{Corresponding log-prior slopes}
The log-prior slope estimates have a high subject to subject variability for conditions A3-A5 (Figure~\ref{fig:model}f-h) compared to conditions A1-A2 (Figure~\ref{fig:model}e). The high inter-subject variability is expected in speed discrimination tasks and in the case of conditions A4-A5, this is particularly magnified by two subjects that have an extremely low value for $a_z$ (their small markers are not visible in Figure~\ref{fig:model}e-f). 

\begin{figure}
\begin{center}
\makebox[\textwidth]{\includegraphics[scale=0.36]{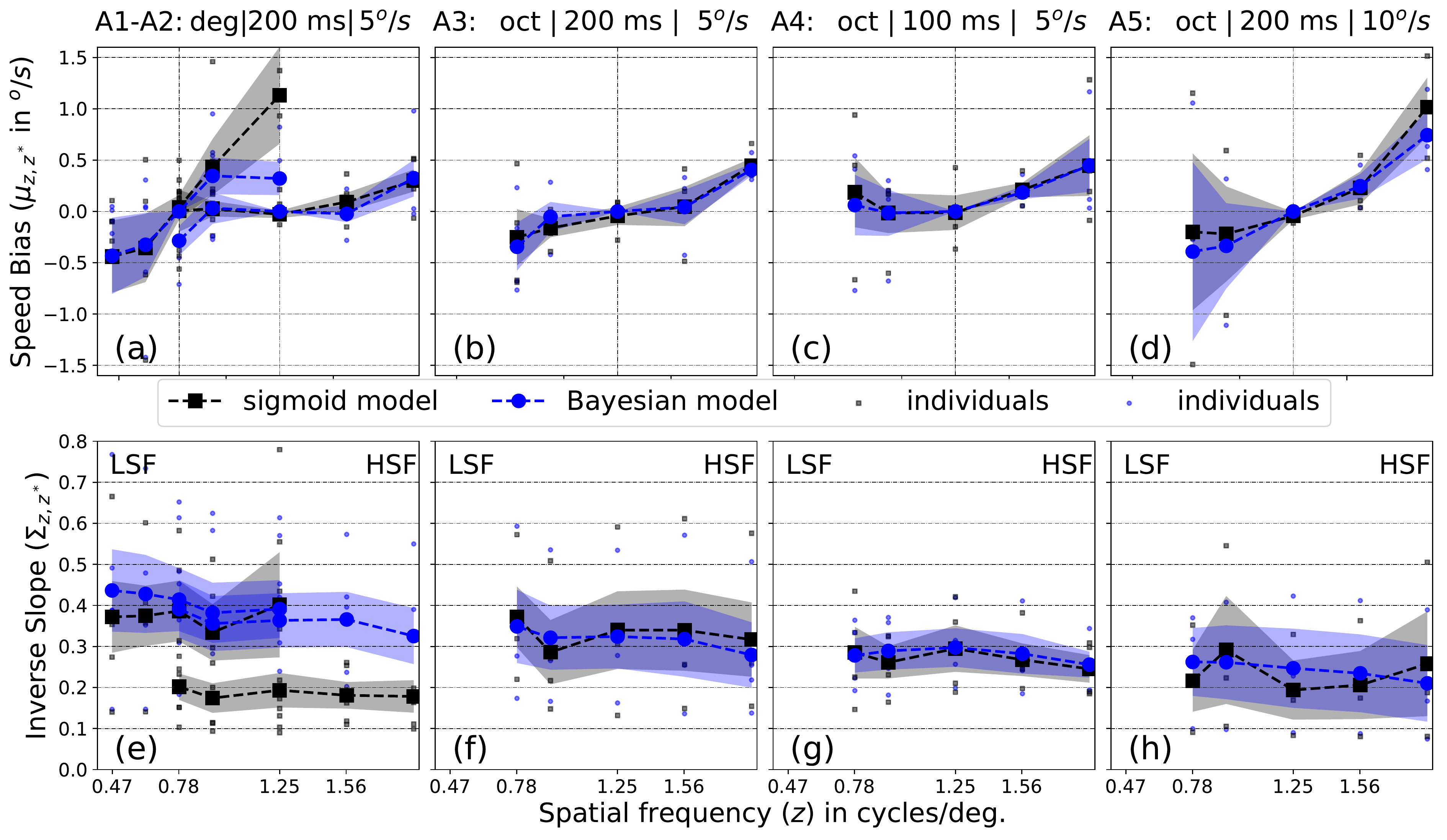}}
\caption{The model fitted speed biases (top row) and inverse sensitivity (bottom row) for the different conditions for the Bayesian model (blue) and the sigmoid model (black).  \textbf{(a-d)} \textit{Speed biases generally increase with increasing spatial frequency.} \textbf{(e-h)} \textit{Inverse sensitivity tend to decrease for the Bayesian model but are configured not to do so for the sigmoid model.} The parameters are indicated above, respectively: bandwidth in octave (oct.) or degree (deg.), value of stimulus lifetime $t^\star$ and reference speed $v^\star$.  Small markers represent individual results and large markers represent population average.  From left to right: conditions A1-A2, A3, A4 and A5.}
\label{fig:psych}
\end{center}
\end{figure}
\begin{figure}
\begin{center}
\makebox[\textwidth]{\includegraphics[scale=0.36]{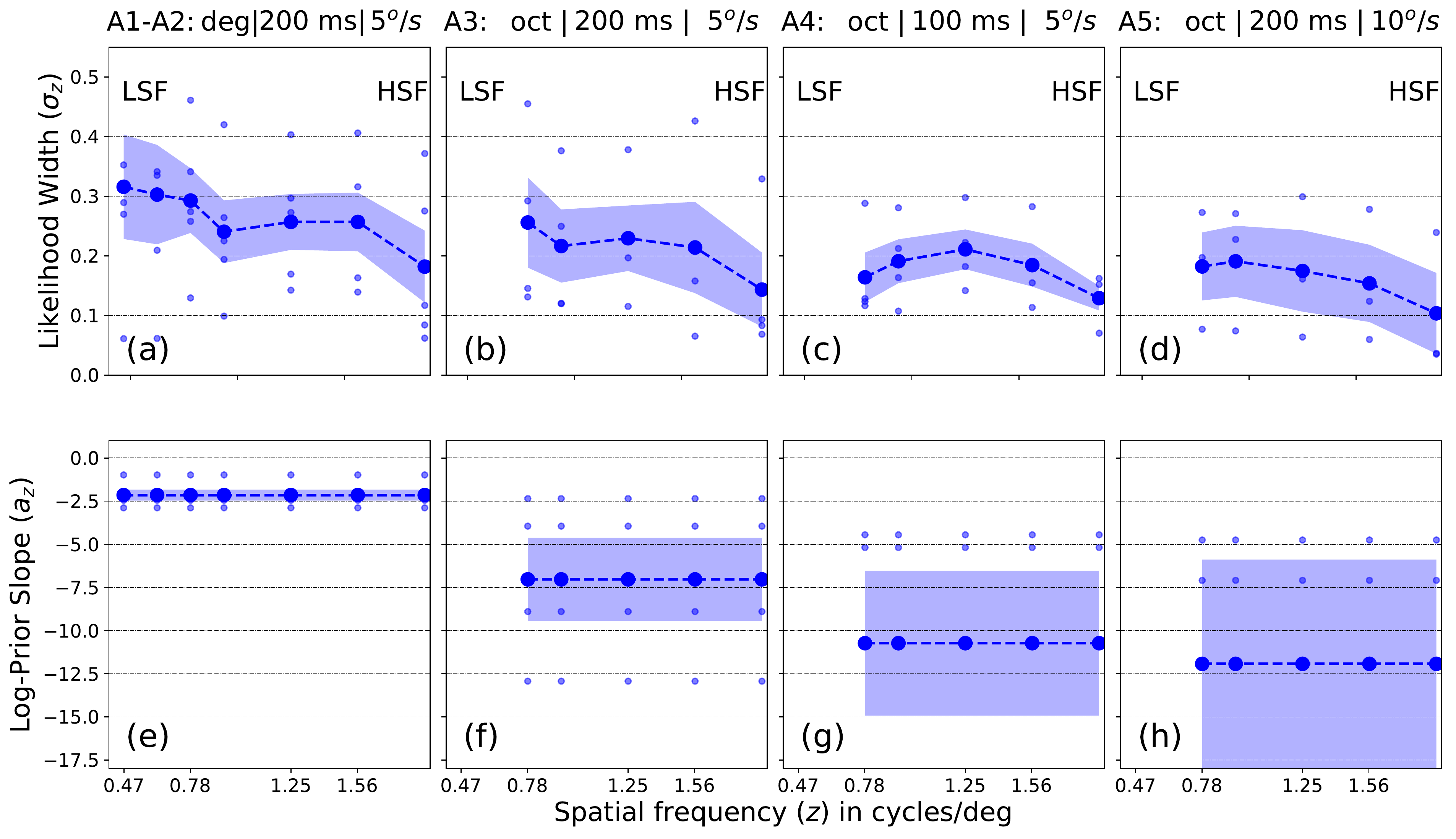}}
\caption{The model fitted likelihood widths (top row) and log-prior slopes (bottom row) for the different conditions for the Bayesian model. \textbf{(a-d)} \textit{Likelihood widths tend to decrease with increasing spatial frequency.} \textbf{(e-h)} \textit{Log-prior slopes are negative and higly variable between subjects.} The parameters are indicated above, respectively: bandwidth in octave (oct.) or degree (deg.), value of stimulus lifetime $t^\star$ and reference speed $v^\star$.  Small markers represent individual results and large markers represent population average. From left to right: conditions A1-A2, A3, A4 and A5.}
\label{fig:model}
\end{center}
\end{figure}
\begin{figure}
\begin{center}
\makebox[\textwidth]{\includegraphics[scale=0.4]{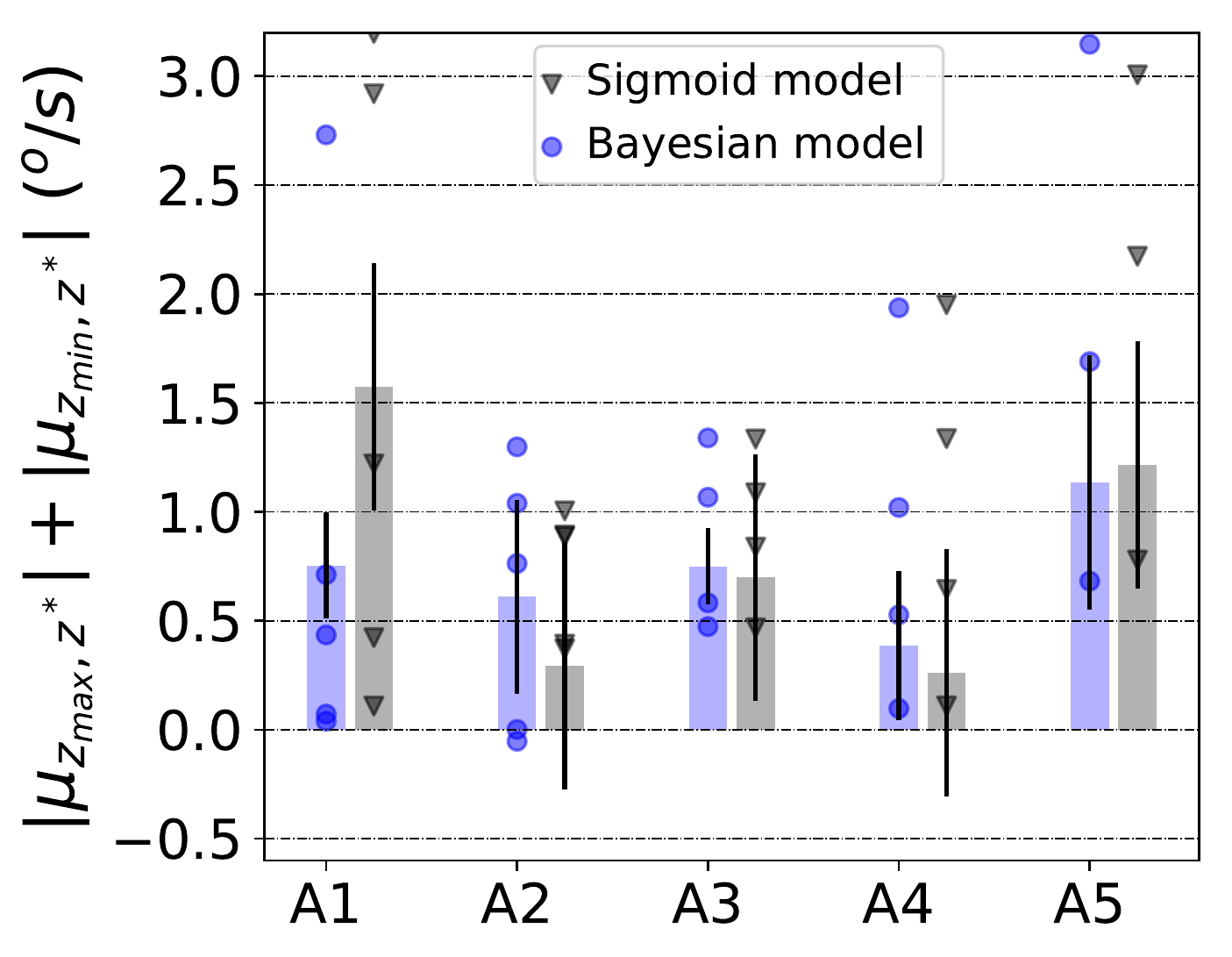}}
\caption{\textit{Biases amplitude.} Sum of the absolute biases at lower and higher spatial frequency averaged over participants. Small markers represent individual results, bars represent population average and error bars represent one standard error of the mean.}
\label{fig:biases-ampli}
\end{center}
\end{figure}

\subsection{Insights into Human Speed Perception}
We exploited the principled and ecologically motivated parameterization of MC to study biases in human speed judgements under a range of parametric conditions. Primarily, we consider the effect of scene scaling on perceived speed, manipulated via central spatial frequencies in a similar way to previous experiments which have shown spatial frequency induced perceived speed biases~\citep{brooks2011contrast,Smith90}. In general, our experimental result confirms that higher spatial frequencies are consistently perceived to be moving faster than compared lower frequencies; which is the same result as reported in a previous study using both simple gratings and compounds of paired gratings, the second of which can be considered as a relatively broadband bandwidth stimulus~\citep{brooks2011contrast} compared to single grating stimuli, without considering the inhibitive interactions we know to occur when multiple gratings are superimposed~\citep{Priebe03}. In that work, they noted that biases were present, but slightly reduced in the compound (broadband) stimuli. That conclusion is consistent with a more recent psychophysics manipulation in which up to four distinct composite gratings were used in relative speed judgements: Estimates were found to be closer to veridical as bandwidth was increased by adding additional components from the set of four, but increasing spatial frequencies generally biased towards faster perceived speed, even if individual participants showed different trends~\citep{jogan2015signal}. Indeed, findings from primate neurophysiology studies have also noted that while responses are biased by spatial frequency, the tendency towards true speed sensitivity (measured as the proportion of individual neurons showing speed sensitivity) increases when broadband stimulation is used~\citep{Priebe03,Perrone01}. A model of visual motion sensitivity with a hierarchical framework, selectively reading from and optimally decoding V1 inputs in an MT layer has also been tested. It was found to be consistent with human speed sensitivity to natural images~\citep{burge2015optimal}.   
        
It is increasingly being recognized that linear systems approaches to interrogating visual processing with single sinusoidal luminance grating inputs represents a powerful but limited approach to study speed perception, as they fail to capture the fact that naturalistic broadband frequency distributions may support speed estimation~\citep{brooks2011contrast,meso2014towards,meso2009speed,gekas2017normalization}. A linear consideration for example may not fully account for the fact that estimation in the presence of multiple sinusoidal components results in linear optimal combination perfoming best among alternatives~\citep{jogan2015signal}. In that case, the simple monotonic increase in perceived speed predicted by the optimal model when additional components were added to the compound is not seen in the data particularly in the difference between 3 and 4 components. This may be due to interaction between components which are not fully captured by this optimal linear model. The current work seeks to extend the body of previous studies by looking at spatial-frequency-induced biases, using a parametric configuration in the form of the Motion Clouds, which allow a manipulation across a continuous scale of frequency and bandwidth parameters. The effect of frequency interactions across the broadband stimulus defined along the two dimensional spatio-temporal luminance plane allows us to measure the perceptual effect of the projection of different areas (e.g. see Figure~\ref{fig:rpt3D}) onto the same speed line. The measurement would be the result of proposed inhibitory interactions, which occur during spatio-temporal frequency integration for speed perception~\citep{Simoncini12,gekas2017normalization}, which cannot be observed with component stimuli separated by several octaves~\citep{jogan2015signal}. 

We use a slower and a faster speed because previous work using sinusoidal grating stimuli have shown that below the slower range ($< 8~\si{\degree / \second}$ ), uncertainty manipulated through lower contrasts causes an under estimation of speeds while at faster speeds ($>16~\si{\degree / \second}$) it causes an overestimation, an effect which itself is not fully explained by Bayesian models with a prior encouraging slow speeds.~\citep{thompson2006speed,hassan2015perceptual}. Our findings show that biases are larger at the faster speed than the slower one. Biases are also generally lower for the octave-controlled than for the cycle-controlled stimuli, indicating that the underlying system was better at processing the former. 

The Bayesian fitting identifies a decrease in the likelihood width estimates which could explain the biases in over half of our fitted psychometric functions. For cases of the same frequency range where comparable likelihoods are estimated, some conditions -- like the low and high $t^\star$ cases -- have very different prior estimates. This result can be interpreted in light of recent work~\citep{gekas2017normalization}: biases might act along the speed line and an orthogonal scale line within the spatio-temporal space, depending on the spread or bandwidth of the stimulus. While the current work does not resolve some of the ongoing gaps in our understanding of speed perception mechanisms, particularly as it does not tackle contrast-related biases, it shows that known frequency biases in speed perception also arise from orthogonal spatial and temporal uncertainties when RMS contrast is controlled - as it is within the MC stimuli. Bayesian models such as the one we apply, which effectively project distributions in the spatiotemporal plane onto a given speed line in which a linear low speed prior applies~\citep{Stocker06} may be insufficient to capture the effect of spatio-temporal priors which may underlie some of the broad set of empirical results. Individual differences which are pervasive in these experiments may also be associated with internal assumptions which can be considered as priors. Indeed for Bayesian models to fully predict speed perception with more complex or composite stimuli, they often require various elaborations away from the simplistic combination of likelihood and low speed prior~\citep{hassan2015perceptual,gekas2017normalization,jogan2015signal,SotiropoulosVR}. Indeed even imaging studies considering the underlying mechanisms fail to find definitive evidence for the encoding of a slow speed prior~\citep{vintch2014cortical}.     

%% file: NECO-07-17-2918-Source-sections/NECO-07-17-2918-Source-sec-discussion.tex

\section{Conclusions}
In this work,  we have proposed and detailed a generative model for the estimation of the motion of dynamic images based on a formalization of small perturbations from the observer's point of view and parameterized by rotations, zooms and translations. We connected these transformations to descriptions of ecologically motivated movements of both observers and the dynamic world. The fast synthesis of naturalistic textures optimized to probe motion perception was then demonstrated, through fast GPU implementations applying auto-regression techniques with much potential for future experimentation. This extends previous work from~\textcite{Leon12} by providing an axiomatic formulation. Finally, we used the stimuli in a psychophysical task and showed that these textures allow one to further understand the processes underlying speed estimation. We used broadband stimulation to study frequency-induced biases in visual perception, using various stimulus configurations including octave bandwidth and RMS contrast-controlled manipulations, which allowed us to manipulate central frequencies as scale invariant stimulus zooms. We showed that measured biases under these controlled conditions were qualitatively the same at both a faster and a slower tested speed. By linking the stimulation directly to the standard Bayesian formalism, we demonstrated that the sensory representation of the stimulus (the likelihoods) in such models is compatible with the generative MC model in over half of the collected empirical data cases. Together with a slow speed prior, the inference framework correctly accounts for the observed bias. We foresee that more experiments with naturalistic stimuli such as MCs and a consideration of more generally applicable priors will be needed in future.

\section*{Acknowledgments}

We thank Guillaume Masson for fruitful discussions during the project. We also acknowledge Manon Bouy\'e and \'Elise Amfreville for proofreading. Finally, we strongly thank the reviewers for their useful comments that helped us improve the readability and quality of this paper.
LUP was supported by EC FP7-269921, ``BrainScaleS'' and  BalaV1 ANR-13-BSV4-0014-02.
The work of JV and GP was supported by the European Research Council (ERC project SIGMA-Vision). 
AIM and LUP were supported by SPEED ANR-13-SHS2-0006.

%% file: NECO-07-17-2918-Source-sections/NECO-07-17-2918-Source-sec-appendix.tex
\appendix
\section{Proofs of Propositions}

In the following, we give the proofs of propositions~\ref{prop-mc-spectrum}, \ref{prop-spde-translation}, \ref{prop-spde-equiv} and \ref{prop:psycurve} introduced in the main article.

\begin{proof}[Proof of Proposition~\ref{prop-mc-spectrum}]
\label{proof:prop-mc-spectrum}
We recall the expression~\eqref{eq-cov-proposition} of the covariance
\eql{\label{eq-cov-proposition-ap}
		\foralls (x,t) \in \RR^3, \quad
		\ga(x,t) = \iiint_{\RR^2} c_{g_\sigma}(\phi_\geom(x-\speed t))  \fv(\speed) \falpha(\geom) \d \speed \d \geom
} 
We denote $(\th,\phi,\z,r) \in \Ga = [-\pi,\pi)^2 \times \RR_+^2$ the set of parameters. 
Denoting $h(x,t) = c_{g_\sigma}( \z R_\th(x - \speed t ) )$, one has, in the sense of generalized functions (taking the Fourier transform with respect to $(x,t)$):
\eq{
	\hat h(\xi,\tau) = \z^{-2} \hat g_\sigma( \z^{-1} R_{\th}(\xi) )^2 \de_{\Vv}(\speed)
	\qwhereq
	\Vv = \enscond{\speed  \in \RR^2 }{ \tau + \dotp{\xi}{\speed} = 0 }.
}
Taking the Fourier transform of equation~\eqref{eq-cov-proposition-ap} and using this computation, the result is that $\hat \ga(\xi,\tau)$ is equal to
\begin{align*}
	\int_\Ga 
	\frac{|\hat g_\sigma\pa{ \z^{-1} R_{\th}(\xi) }|^2}{\z^2} 
	\de_{\Vv}(\vz + r(\cos(\phi),\sin(\phi))) \ftheta(\th) \fz(\z) \fr(r)  
	\: \d\th \d\z \d r \d\phi.
\end{align*}
Therefore when $\sigma \rightarrow 0$, one has in the sense of generalized functions
\eq{
	|\hat g_\sigma\pa{ \z^{-1} R_{\th}(\xi) }|^2 \rightarrow \de_{\Bb}(\th,\z)
	\qwhereq
	\Bb = \enscond{(\th,\z) }{ \z^{-1} R_{\th}(\xi) = \xiz }.
}
Observing that
$\de_{\Vv}(\speed) \de_{\Bb}(\th,\z) = \de_{\Cc}(\th,\z,r)$ where
\eq{
	\Cc = \enscond{(\th,\z,r)}{
		\z = \norm{\xi}, \; 
		\th =  \angle{\xi}, \; 
		r = -\frac{\tau}{ \norm{\xi}\cos(\angle{\xi}-\phi) } - \frac{\vzMod \cos(\angle{\xi} - \vzAng)}{\cos(\angle{\xi}-\phi)}
	}
}
one obtains the desired formula.
\end{proof}

\begin{proof}[Proof of Proposition~\ref{prop-spde-translation}]
\label{proof:prop-spde-translation}
	This follows by computing the derivative in time of the warping equation~\eqref{eq-time-warping}, denoting $y \eqdef x+\vz t$
	\begin{align*}
		\partial_t I_0(x,t) &= \partial_t I(y,t) + \dotp{\nabla I(y,t)}{\vz}, \\
		\partial_t^2 I_0(x,t) &= \partial_t^2 I(y,t) + 2\dotp{\partial_t\nabla I(y,t)}{\vz}
			+ \dotp{ \partial_x^2 I(y,t) \vz }{\vz} 
	\end{align*}
	and plugging this into equation~\eqref{eq-spde} after remarking that the distribution of $\pd{W}{t}(x,t)$ is the same as the distribution of $\pd{W}{t}(x-\vz t, t)$. 
\end{proof}

\begin{proof}[Proof of Proposition~\ref{prop-spde-equiv}]
\label{proof:prop-spde-equiv}
	\newcommand{\MC}{\text{\tiny MC}}
	For this proof, we denote $I^\MC$ the motion cloud defined by equation~\eqref{eq-dfn-mc-spectrum}, and $I$ a stationary solution of the sPDE defined by equation~\eqref{eq-spde}, which exists according to Theorem~\ref{thm-sol-eq} because $\hsiW^2 \hat\nu^3 \in L^1$. Indeed $\fz$ and $\ftheta$ are probability distributions and $\xi \mapsto \frac{1}{\norm{\xi}^2}$ does not change the continuity at $0$. We aim to show that under the specification~\eqref{eq-equiv-spde-mc}, they have the same covariance. This is equivalent to showing that $I_0^\MC(x,t) =  I^\MC(x+ct,t)$ has the same covariance as $I_0$. 
	For any fixed $\xi$, equation~\eqref{eq-spde-freq} admits a unique stationary solution $\hat I_0(\xi,\cdot)$ (Theorem \ref{thm-sol-eq}) which is a stationary Gaussian process of zero mean and with a covariance which is $\hsiW^2(\xi)  r \star \bar r$, where $r$ is the impulse response (i.e. formally taking $a=\de$) of the ODE $r'' + 2r'/u + r/u^2 = a$, where we denoted $u=\hat\nu(\xi)$. This impulse response can be shown to be $r(t) = t e^{-t/u} \one_{\RR^+}(t)$. The covariance of $\hat I_0(\xi,\cdot)$ is thus, after some computation, equal to $\hsiW^2(\xi)  r \star \bar r = \hsiW^2(\xi) h(\cdot/u)$, where $h(t) = \frac{u^3}{4}(1+|t|)e^{-|t|}$.
	Taking the Fourier transform of this equality, the power spectrum $\hat\ga_0$ of $I_0$ thus reads
	\eql{\label{eq-powerspec-spde}
		\hat \ga_0(\xi,\tau) = \frac{1}{4}\hsiW^2(\xi) \hat\nu(\xi)^3 \tilde h(\hat\nu(\xi) \tau)
		\qwhereq
	\tilde h(s) = \frac{1}{(1+s^2)^2}
	}
	and where it should be noted that this function $h$ is the same as the one introduced in equation~\eqref{eq-expression-fr}.
	%
	The covariance $\ga^\MC$ of $I^\MC$ and $\ga_0^\MC$ of $I_0^\MC$ are related by the relation 
	\eql{\label{eq-powerspec-MC}		
		\hat \ga_0^{\MC}(\xi,\tau) = \hat \ga^{\MC}(\xi,\tau-\dotp{\xi}{\vz}) = 
		\frac{ 1 }{\norm{\xi}^2} \fz( \norm{\xi} )
		\ftheta\pa{ \angle{\xi}  } 
		\hat h\pa{-\frac{\tau }{  \sr \norm{\xi} } }.
	}
	where we used the expression~\eqref{eq-dfn-mc-spectrum} for $\hat \ga^{\MC}$ and the value of $\Ll(\fr)$ given by equation~\eqref{eq-expression-fr}.
	Condition~\eqref{eq-equiv-spde-mc} guarantees that expression~\eqref{eq-powerspec-spde} and~\eqref{eq-powerspec-MC} coincide, and thus $\hat \ga_0=\hat \ga_0^{\MC}$.
\end{proof}

\begin{proof}[Proof of Proposition~\ref{prop:psycurve}]
\label{proof:prop-psycurve}
One has the closed form expression for the MAP estimator
\eq{
	\hat \q(\m) = \m + a_\nuis \sigma^2_\nuis, 
}
and hence, denoting $\Nn(\mu,\sigma^2)$ the Gaussian distribution of mean $\mu$ and variance $\sigma^2$, 
\eq{
	\hat\q(\M_{\q,\nuis}) \sim \Nn(\q+a_\nuis \sigma^2_\nuis, \sigma^2_\nuis)
}
where $\sim$ means equality of distributions.
One thus has 
\eq{
	\hat\q(\M_{\q,\nuis^\star}) - \hat\q(\M_{\q^\star,\nuis})
		\sim
		\Nn( \q-\q^\star + a_{\nuis^\star} \sigma_{\nuis^\star}^2  - a_{\nuis} \sigma_{\nuis}^2, 
		\sigma_{\nuis^\star}^2 + \sigma_{\nuis}^2 ),
} 
which leads to the results by taking expectation. 
\end{proof}

\section{Complementary Results}
\label{ap:comp-res}

In the following, we states two complementary results regarding our sPDE formulation of MCs. Theorem~\ref{thm-sol-eq} gives necessary conditions for equation~\eqref{eq-spde} to have solutions. More general results can be found in~\textcite{UnserBook}.

\begin{thm}
\label{thm-sol-eq}
	If $(\hat \al, \hat \be)$ are non-negative and $ \frac{\hsiW^2 }{\hat\al \hat\be} \in L^1$, then equation~\eqref{eq-spde} has a unique causal and stationary solution, \textit{ie} uniquely defines the distribution of a sPDE cloud.
\end{thm}
\begin{proof} 
Consider equation~\eqref{eq-spde-freq}, the Fourier transform of equation~\eqref{eq-spde}, which has causal and stationary solutions (see the general case of Levy-driven sPDE, Theorem 3.3 in~\textcite{brockwell2009existence}). Given that $ \frac{\hsiW }{\hat\al \hat\be} \in L^1 $, these solutions have an integrable spatial power spectrum. Then, one could take their inverse Fourier transform and get the solution, which is unique by construction.
\end{proof}
\begin{rem} 
There are different ways to define uniqueness of solution for sPDE. In Theorem~\ref{thm-sol-eq}, uniqueness has to be understood in terms of sample path, meaning that two solutions $(X, \tilde X)$ of equation~\eqref{eq-spde} verify $\PP(\forall t \in \RR, X_t = \tilde X_t)=1$. This notion of uniqueness is strong and it implies uniqueness in distribution \textit{ie.} that $X$ and $\tilde X$ have the same law. 
\end{rem}

The following proposition gives a closed-form expression for the function $\Ll^{-1}(h)$ where $\Ll$ is defined in equation~\eqref{eq:linear-trans-spe} and $h$ is defined in equation~\eqref{eq-expression-fr}. In particular, we show that it is a function in $L^1(\RR)$, \textit{ie} it has a finite integral, which can be normalized to unity to define a density distribution. 
Figure~\ref{fig:Ltransform} shows a graphical display of this distribution.
\begin{prop}
\label{prop:inverse}
One has
\eq{
	\Ll^{-1}(h)(u) = \frac{2-u^2}{\pi(1+u^2)^2}-\frac{u^2(u^2+4)(\log(u)-\log(\sqrt{u^2+1}+1))}{\pi(u^2+1)^{5/2}}.
}
In particular, one has
\eq{
	\Ll^{-1}(h)(0) = \frac{2}{\pi}
	\qandq
	\Ll^{-1}(h)(u) \sim \frac{1}{2\pi u^3}
	\quad\text{when}\quad
	u \rightarrow +\infty.
}
\end{prop}

\begin{proof}
 The variable substitution $x=\cos(\phi)$ can be used to rewrite~\eqref{Ltransform} as $$ \forall u \in \RR,  \quad \Ll(h)(u) = 2\int_{0}^{1} h\left(-\frac{u}{x} \right) \frac{x}{\sqrt{1-x^2}} \frac{\d x}{x}.$$ 
In such a form, we recognize a Mellin convolution which could be inverted by the use of the Mellin convolution table~\citep{oberhettinger2012tables}.
\end{proof}

\begin{figure}[!ht]
\begin{center}
\includegraphics[scale=0.5]{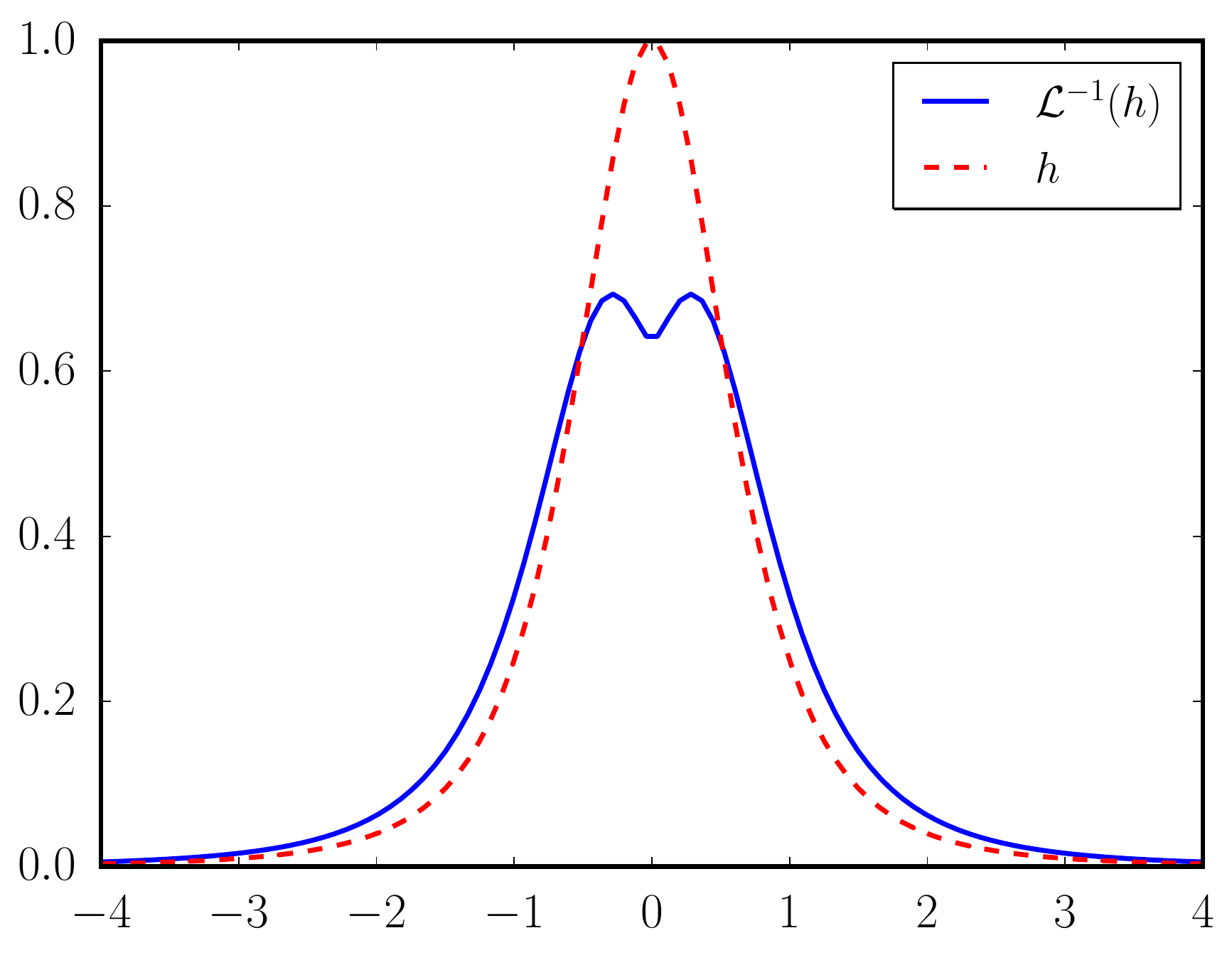}
\caption{Functions $h$ and $\Ll^{-1}(h)$.}
\label{fig:Ltransform}
\end{center}
\end{figure}

\section{Alternative Parametrization of Zoom Distribution}
\label{sec:appendix}

In practice, the parameters $(\tlnz, \tlnsz)$ are not convenient to manipulate because they have no direct physical meaning.
\paragraph{Parametrization of $\fz$ by mode and standard deviation}
We propose a more intuitive parametrization using mode and standard deviation $(\lnz,\lnsz)$ 
\eq{
	\lnz \eqdef \argmax_{\z} \fz(\z)
	\qandq
	\lnsz^2 \eqdef \EE( \Z^2) - \EE(\Z)^2 . 
}
Once $(\lnz,\lnsz)$ are fixed, it is easy to compute the corresponding $(\tlnz,\tlnsz)$ to plug into expression~\eqref{eq-biosinspired-fz}, simply by solving a polynomial equation~\eqref{eq-relation-mode}, as detailed in the following proposition.   

\begin{prop}
One has
\eq{
	\lnz = \frac{\tlnz}{1+\tlnsz^2} \qandq \lnsz^2 = \tlnz^2 \tlnsz^2 (1+\tlnsz^2).
} 
Such a formula can be inverted by finding the unique positive root of 
\eql{\label{eq-relation-mode}
	\tlnsz^2 (1+\tlnsz^2)^3 - \frac{\lnsz^2}{\lnz^2} = 0 
	\qandq \tlnz = \lnz(1+\tlnsz^2).
}
\end{prop}

\begin{proof}
The primary relations are established using standard calculations from the probability density function $\fz$~\citep{johnson1994continuous}. The relations~\eqref{eq-relation-mode} follow standard arithmetic. 

\end{proof}

\paragraph{Parametrization of $\fz$ by mode and octave bandwidth}

Differences in perception are often more relevant in a log domain, therefore it is useful to parametrize  $\fz$ by its mode $\lnz$ and octave bandwidth $\lnbw$, which is defined by 
\eq{
	\lnbw \eqdef \frac{ \ln\left( \frac{z_+}{z_- }\right) }{ \ln(2) }
}
where $(z_-,z_+)$ are respectively the successive half-power cutoff frequencies, which verify $ \fz(z_-) = \fz(z_+) = \frac{\fz(\lnz)}{2} $ with $ z_- <  z_+$.
\begin{prop}
\label{prop-bandwidth-var}
One has
\eql{\label{eq-relation-bw}
	\lnbw = \sqrt{\frac{8\ln(1+\tlnsz^2)}{\ln(2)}} 
	\quad \text{and conversely}\quad
	\tlnsz = \sqrt{\exp\left(\frac{\ln(2)}{8} \lnbw^2 \right) - 1}.
}

\end{prop}
\begin{proof}
Using the fact that $ \fz(z_-) = \fz(z_+) = \frac{\fz(\lnz)}{2} $, one shows that $ X_+ = \ln \left( \frac{z_+}{\lnz} \right) $ and $ X_- = \ln \left( \frac{z_-}{\lnz} \right) $ are the two roots of the following polynomial (with $ X_- \le  X_+$).
\eq{
	Q(X) = X^2 + 2 \ln(1+\tlnsz^2) X - 2 \ln(2) \ln(1+\tlnsz^2) + \frac{1}{2} \ln(1 + \tlnsz^2)^2 
}
This allows to compute $\lnbw$.
\end{proof}
Using Proposition~\ref{prop-bandwidth-var}, it is possible to obtain the parametrization of bandwidth prevalent in manipulations used in visual psychophysics experiments.